\documentclass{article}
\usepackage[top=1in, bottom=1in, left=1in, right=1in]{geometry}
\usepackage[utf8]{inputenc}

\usepackage{overpic}
\usepackage[font=small,labelfont=bf]{caption}
\usepackage{dsfont}
\usepackage{graphicx}
\usepackage{algorithm}
\usepackage{algpseudocode}
\usepackage{mathtools}
\graphicspath{ {./figs/} }

\usepackage{url}

\usepackage{microtype}
\usepackage{graphicx}
\usepackage{booktabs} %
\usepackage{hyperref}

\usepackage{caption}
\usepackage{amsmath,amssymb,amsfonts,amsthm}

\usepackage{mathrsfs}
\usepackage{xcolor}
\usepackage{mathtools}
\usepackage{dsfont}
\usepackage{hyperref}
\usepackage{bm}
\usepackage{nicefrac}
\usepackage{wrapfig}
\usepackage{lipsum}
\usepackage{enumerate}

\newcounter{assumption}%
\renewcommand{\theassumption}{\arabic{assumption}}

\newtheorem{theorem}{Theorem}
\newtheorem{corollary}{Corollary}
\newtheorem{proposition}{Proposition}
\newtheorem{remark}{Remark}
\newtheorem{definition}{Definition}

\numberwithin{lemma}{section} 
\numberwithin{theorem}{section} 
\numberwithin{corollary}{section} 
\numberwithin{proposition}{section} 
\numberwithin{definition}{section} 
\numberwithin{example}{section} 
\numberwithin{question}{section} 

\DeclareMathOperator*{\argmin}{argmin}

\usepackage{multirow}
\usepackage{caption}
\usepackage{subcaption}

\usepackage[style=numeric,sorting=none]{biblatex}
\usepackage{comment}

\usepackage{amsmath}      
\usepackage{amssymb}      
\usepackage{amsthm}       
\usepackage{enumerate}    
\usepackage{booktabs}     
\usepackage{array}

\title{Information theory and discriminative sampling for model discovery}
\author{Yuxuan Bao$^*$, 
J. Nathan Kutz$^{*,\dag}$\\[.1in]$^*$Department of Applied Mathematics, University of Washington, Seattle, USA\\
$^\dag$Department of Electrical and Computer Engineering, University of Washington, Seattle, USA}

\begin{document}

\maketitle

\begin{abstract}
Fisher information and Shannon entropy are fundamental tools for understanding and analyzing dynamical systems from complementary perspectives. They can characterize unknown parameters by quantifying the information contained in variables, or measure how different initial trajectories or temporal segments of a trajectory contribute to learning or inferring system dynamics. In this work, we leverage the Fisher Information Matrix (FIM) within the data-driven framework of {\em sparse identification of nonlinear dynamics} (SINDy). We visualize information patterns in chaotic and non-chaotic systems for both single trajectories and multiple initial conditions, demonstrating how information-based analysis can improve sampling efficiency and enhance model performance by prioritizing more informative data. The benefits of statistical bagging are further elucidated through spectral analysis of the FIM. We also illustrate how Fisher information and entropy metrics can promote data efficiency in three scenarios: when only a single trajectory is available, when a tunable control parameter exists, and when multiple trajectories can be freely initialized. As data-driven model discovery continues to gain prominence, principled sampling strategies guided by quantifiable information metrics offer a powerful approach for improving learning efficiency and reducing data requirements.
\end{abstract}

\section{Introduction}

The data-driven discovery of underlying ordinary differential and partial differential equations (ODEs and PDEs) characterizing complex systems remains a grand challenge in modern science and engineering. In comparison with analytic approaches that often require prior knowledge or first principles of the underlying dynamical system, data-driven approaches are capable of extrapolating behavior directly from observational data. And among the numerous data-driven methods for model discovery, the {\em sparse identification of nonlinear dynamics} (SINDy) \cite{brunton2016discovering} has become increasingly popular for its ability to efficiently learn interpretable dynamical systems. SINDy utilizes a sparsity-promoting regression on a library of candidate functions, and fits measurement data with the fewest terms possible.  Thus a parsimonious model of the underlying dynamics is learned.  SINDy has evolved significantly since its inception, with various modifications of the algorithm aimed at reducing data-reliance, promoting noise-robustness, learning implicit functional forms \cite{mangan2017inferring}, accommodating rational function nonlinearities and control inputs \cite{kaheman2020sindy}, robustification using ensembles from subsample aggregation \cite{delahunt2022toolkit,reinbold2021robust} or a weak formulation~\cite{messenger2021weak}, and promoting stability by producing bounded trajectories~\cite{kaptanoglu2021promoting}. Recent efforts have also evolved beyond traditional bagging methods \cite{breiman1996bagging} to the utilization of ensemble statistics \cite{fasel2022ensemble}, thus further improving the robustness on noisy conditions with low data limit, which works especially well on spatio-temporal systems.  To date, such improvements have focused on improvements and generalization of the core SINDy algorithm itself, and their various contributions are included in the open source pySINDY package~\cite{kaptanoglu2021pysindy}. In contrast, this work shifts the focus from algorithmic refinements to the data itself. In particular, standard SINDy formulations implicitly assume that all observations contribute equally to model discovery. We demonstrate that this assumption is generally invalid: different segments of a temporal trajectory carry substantially different amounts of information about the underlying dynamics. This heterogeneity can be rigorously quantified using Fisher information and entropy–based metrics. As a consequence, distinct portions of a trajectory contribute unequally to the identification of governing ODE or PDE systems. Exploiting this observation enables more efficient data utilization, as accurate model identification can be achieved using only the most informative subsets of the data. Moreover, this perspective provides deeper physical insight into the dynamics by highlighting the specific temporal regimes that are most diagnostic of the underlying governing laws.





In this work, we experimentally demonstrate how different trajectories, and different temporal segments of a single trajectory, have distinct impacts on data-driven model performance. Over trials on both chaotic and non-chaotic systems with numerous initial conditions, we visualize their information patterns and propose FIM-based metrics that optimize sampling efficiency. To further demonstrate how sampling efficiency can be improved through information analysis, we provide concrete applications of Fisher information and entropy metrics for promoting data efficiency in three distinct scenarios: when only one trajectory is available, when a control parameter is available for tuning, and when multiple trajectories are available with freely chosen initial conditions. The efficacy of bootstrap aggregation (bagging) methods can also be explained through spectral analysis of the Fisher Information Matrix - as more trajectory samples are aggregated, the FIM spectrum becomes less skewed and the parameter space is constrained in more directions, thereby improving model performance. In model discovery problems, especially as systems become increasingly complex, the information gained from different initial conditions or different temporal segments of a trajectory can differ drastically.\cite{vasey2025influence} Such discrepancies are reflected in the FIM eigenvalue distribution and ultimately affect the recovery of governing equations.

\section{Preliminaries}\label{sec:preliminaries}

This section introduces the foundational concepts employed throughout this work: the Sparse Identification of Nonlinear Dynamics (SINDy) framework for data-driven model discovery, the Fisher Information Matrix (FIM) as a measure of parameter identifiability, and entropy-based metrics from the Shannon family for quantifying regularity in time series data.

\subsection{Sparse Identification of Nonlinear Dynamics (SINDy)}\label{subsec:sindy}

SINDy \cite{brunton2016discovering} is a data-driven framework developed to identify parsimonious models of nonlinear partial differential equations in the form
\begin{equation}\label{eq:pde_form}
    \mathbf{u}_t = \mathbf{F}(\boldsymbol{\theta}, \mathbf{u}, \mathbf{u}_x, \mathbf{u}_{xx}, \ldots, x),
\end{equation}
where $\mathbf{F}(\cdot)$ is a system of nonlinear functions of the state $\mathbf{u}$, its partial derivatives, and other parameters $\boldsymbol{\theta}$. The method utilizes sparse regression on measurement data to identify active terms in $\mathbf{F}$ from a library of linear and nonlinear candidates including partial derivatives.

Given input data $\mathbf{D} = [\mathbf{d}_1, \mathbf{d}_2, \ldots, \mathbf{d}_n]^\top$, the algorithm \cite{rudy2017data} computes a library of candidate terms, where derivatives are obtained using finite difference or interpolation methods depending on the noise level in measurements. This generates an evaluated library matrix
\begin{equation}\label{eq:library_matrix}
    \boldsymbol{\Theta}(\mathbf{D}) = \begin{bmatrix} \mathbf{1} & \mathbf{D} & \mathbf{D}^2 & \cdots & \mathbf{D}_x & \mathbf{D}_{xx} & \cdots & \mathbf{D}\mathbf{D}_x & \cdots \end{bmatrix}^\top.
\end{equation}
With the time derivative represented in vectorized form, the original PDE can be written as
\begin{equation}\label{eq:sindy_regression}
    \mathbf{D}_t = \boldsymbol{\Theta}(\mathbf{D})\boldsymbol{\Xi},
\end{equation}
where $\boldsymbol{\Xi}$ is sparse in most practical scenarios. The identification of parameters thus reduces to the optimization problem
\begin{equation}\label{eq:sindy_opt}
    \hat{\boldsymbol{\Xi}} = \argmin_{\boldsymbol{\Xi}} \frac{1}{2}\|\boldsymbol{\Theta}(\mathbf{D})\boldsymbol{\Xi} - \mathbf{D}_t\|_2^2 + \mathcal{R}(\boldsymbol{\Xi}),
\end{equation}
which can be solved by sparsity-promoting regression methods such as sequential thresholding ridge regression (STRidge). Rudy et al.\ \cite{rudy2017data} demonstrated that STRidge exhibits superior performance when data are highly correlated. To accommodate the amplifying effect of noise in high-order partial derivatives, integral terms and weak formulations \cite{schaeffer2017sparse, messenger2021weak, reinbold2020using} have been introduced to further enhance robustness. Recent efforts have incorporated weak-form estimates with latent space techniques to learn nonlinear lower-dimensional dynamics for high-dimensional data \cite{he2025physics}.

The SINDy framework has evolved significantly since its inception, with modifications aimed at reducing data reliance, promoting noise robustness, learning implicit functional forms \cite{mangan2017inferring}, accommodating rational function nonlinearities and control inputs \cite{kaheman2020sindy}, robustification using ensembles from subsample aggregation \cite{delahunt2022toolkit, reinbold2021robust} or weak formulations \cite{messenger2021weak}, and promoting stability by producing bounded trajectories \cite{kaptanoglu2021promoting}. Recent efforts have evolved beyond traditional bagging methods \cite{breiman1996bagging} to the utilization of ensemble statistics \cite{fasel2022ensemble}, further improving robustness under noisy conditions with limited data, which works especially well on spatio-temporal systems. These contributions are included in the open-source PySINDy package \cite{kaptanoglu2021pysindy}.

\subsection{Fisher Information Matrix (FIM)}\label{subsec:fim_prelim}

The Fisher Information Matrix (FIM) \cite{ly2017tutorial} provides a measure of the asymptotic variability of parameter estimators given the observed data, thereby quantifying the amount of information that the data carries about unknown parameters. It is defined as
\begin{equation}\label{eq:fim_def}
    I_{ij} = \mathbb{E}\left[\frac{\partial}{\partial \theta_i}\log p(\mathbf{X}; \boldsymbol{\theta}) \cdot \frac{\partial}{\partial \theta_j}\log p(\mathbf{X}; \boldsymbol{\theta})\right],
\end{equation}
which yields a positive semidefinite matrix characterizing the information content in the data $\mathbf{X}$ for the estimated parameters $\boldsymbol{\theta}$. When the largest eigenvalue of the FIM is large, the probability density function for $\mathbf{X}$ conditioned on $\boldsymbol{\theta}$ exhibits greater sensitivity to changes in $\boldsymbol{\theta}$, indicating that the data $\mathbf{X}$ carries more information about the parameters. The spectral distribution of the FIM is therefore fundamental in analyzing the quality of input data for parameter estimation.

Our study focuses on the landscape of the FIM spectrum, particularly the largest eigenvalues, within the SINDy framework. Geometrically, the eigenvalue pattern of sub-regions of input data indicates the model's proximity to the observed portions of data, providing insight for identifying more informative trajectory segments for learning the estimated parameters. Both chaotic and non-chaotic systems are considered across a range of noise levels and initial conditions.

\subsection{Entropies and Metrics}\label{subsec:entropies}

The use of entropies as measures of randomness has been widely applied in both stochastic processes and deterministic systems. Entropy metrics can quantify structural similarity between data blocks and even between different dynamical systems. They are particularly useful in chaotic settings where small measurement noise is greatly amplified over short time periods.

\subsubsection{The Kolmogorov--Sinai Entropy}\label{subsec:ks_entropy}

Drawn from the Shannon family, the Kolmogorov--Sinai (KS) entropy \cite{ott2002chaos} serves as a measure of uncertainty in dynamical systems. In the classical setting, both time and spatial grids are partitioned, and the KS entropy is defined as the limiting entropy when partition intervals approach zero:
\begin{equation}\label{eq:ks_entropy}
    h_{\mathrm{KS}}(\mathbf{X}, t) = -\lim_{\delta t \to 0} \lim_{\delta x \to 0} \lim_{k \to \infty} \frac{1}{k\delta t} \sum P(x_1, x_2, \ldots, x_k) \log P(x_1, x_2, \ldots, x_k).
\end{equation}
The KS entropy provides a criterion for defining chaos \cite{ott2002chaos}: a positive KS entropy serves as a clear indication of chaotic behavior in the system. Pesin's identity \cite{eckmann1985ergodic} further establishes a connection between the KS entropy and the Lyapunov exponents, showing that the KS entropy is bounded by the sum of all positive Lyapunov exponents.

\subsubsection{The Approximate Entropy (ApEn)}\label{subsec:apen}

The Approximate Entropy (ApEn), proposed by Pincus \cite{pincus1991approximate}, is a metric that measures spatial and temporal regularities in data series. Heuristically, it quantifies the correlation between data patterns: low ApEn suggests that a system is predictive in its repeating patterns and oscillatory frequencies, while high ApEn indicates independence in the data series, emergence of new patterns, and lower predictability from current observations.

Given a sequence of length $N$, a template length $m$, and a distance threshold $r$, the approximate entropy is calculated as
\begin{equation}\label{eq:apen}
    \mathrm{ApEn}(m, r, N) = \phi^m(r) - \phi^{m+1}(r),
\end{equation}
where
\begin{equation}\label{eq:phi_m}
    \phi^m(r) = \frac{1}{N - m + 1} \sum_{i=1}^{N-m+1} \log C_i^m(r),
\end{equation}
and $C_i^m(r)$ is the proportion of template vectors $\mathbf{u}_m(j) = (u_j, u_{j+1}, \ldots, u_{j+m-1})^\top$ from all $N - m + 1$ choices of $j$ that lie within radius $r$ of the template vector $\mathbf{u}_m(i) = (u_i, u_{i+1}, \ldots, u_{i+m-1})^\top$ under the $\ell^\infty$ norm.

Despite its utility in reflecting data series regularity, ApEn computation is constrained by the fixed template length $m$ and distance threshold $r$. The metric can be highly sensitive to changes in $r$. While ApEn performs well for oscillatory data with consistent frequencies, dynamical systems with varying behaviors require different optimal choices of $m$ and $r$ at different stages, increasing the complexity of hyperparameter tuning.

Moreover, ApEn is a biased metric compared to sample entropy \cite{yentes2013appropriate}. To ensure the logarithm is well-defined, the probability calculation includes the point $\mathbf{u}_m(i)$ itself in computing $C_i^m(r)$, rather than using conditional probability. This bias is amplified by the logarithm, especially when $C_i^m(r)$ is small, corresponding to cases where few or no matches are found for the template vector.

\subsubsection{The Sample Entropy (SampEn)}\label{subsec:sampen}

Unlike ApEn, which suffers from self-comparison bias, the Sample Entropy (SampEn) proposed by Richman and Moorman \cite{richman2000physiological} eliminates this bias by excluding self-matches from the probability estimates.

Given a sequence of length $N$, template length $m$, and distance threshold $r$, the sample entropy is computed by first defining two similarity measures at lengths $m$ and $m+1$:
\begin{align}
    B^m(r) &= \frac{1}{N-m} \sum_{i=1}^{N-m} B_i^m(r) = \frac{1}{N-m} \cdot \frac{1}{N-m-1} \sum_{i=1}^{N-m} \sum_{\substack{j=1 \\ j \neq i}}^{N-m} \mathbf{1}_{d(\mathbf{u}_m(i), \mathbf{u}_m(j)) \leq r}, \label{eq:Bm}\\
    A^m(r) &= \frac{1}{N-m} \sum_{i=1}^{N-m} A_i^m(r) = \frac{1}{N-m} \cdot \frac{1}{N-m-1} \sum_{i=1}^{N-m} \sum_{\substack{j=1 \\ j \neq i}}^{N-m} \mathbf{1}_{d(\mathbf{u}_{m+1}(i), \mathbf{u}_{m+1}(j)) \leq r}, \label{eq:Am}
\end{align}
where $d(\cdot, \cdot)$ denotes the $\ell^\infty$ distance. The sample entropy is then
\begin{equation}\label{eq:sampen}
    \mathrm{SampEn}(m, r, N) = -\log \frac{A^m(r)}{B^m(r)}.
\end{equation}

Both ApEn and SampEn efficiently quantify regularity in dynamical systems, and they converge to the same value as $N \to \infty$. However, SampEn is generally considered superior for the following reasons. First, ApEn includes self-comparison bias \cite{yentes2013appropriate}, which is especially significant when few close matching vectors exist. Second, by averaging over $N - m$ templates, SampEn is statistically independent of data length, whereas ApEn is not.

A limitation shared by both methods is their dependence on the template length $m$ and distance threshold $r$; both entropies can be quite sensitive to these hyperparameters, particularly as a dynamical system becomes increasingly chaotic.

\subsubsection{Entropy Search}\label{subsec:entropy_search}

The use of entropies is common in optimization tasks, with Bayesian optimization being one of the most powerful techniques. Bayesian optimization comprises two key components: a Bayesian surrogate model that makes stochastic approximations from previous evaluations of the objective function, and an acquisition function that efficiently searches for the next evaluation point. These ingredients correspond to two core criteria in optimization: the quality of the predictive model is crucial for robustness to perturbations, and the acquisition function is important for sampling efficiency, especially when evaluations are costly.

In recent years, entropy-based acquisition functions have emerged from traditional data selection methods, selecting evaluation points that maximize information gain or, equivalently, minimize uncertainty about model parameters. In 2012, Hennig and Schuler \cite{hennig2012entropy} proposed the first entropy search acquisition function that maximizes information about the function maximum given new data:
\begin{equation}\label{eq:entropy_search}
    \alpha_{\mathrm{ES}}(x) = H[p(x^* | \mathcal{D})] - \mathbb{E}_{p(y|x,\mathcal{D})}[H[p(x^* | \mathcal{D} \cup \{(x, y)\})]],
\end{equation}
where $H[\cdot]$ denotes the differential entropy, $x^*$ is the optimum location, and $\mathcal{D}$ represents the current data. The first term is the entropy of the posterior over the optimum based on current data (a constant for any incoming point), while the second term accounts for the expected reduction in entropy after observing new data $(x, y)$.

Computing both terms proved challenging until Hern\'andez-Lobato et al.\ \cite{hernandez2016predictive} reformulated the entropy search into an equivalent but more tractable form, termed Predictive Entropy Search (PES):
\begin{equation}\label{eq:pes}
    \alpha_{\mathrm{PES}}(x) = H[p(y | x, \mathcal{D})] - \mathbb{E}_{p(x^*|\mathcal{D})}[H[p(y | x, x^*, \mathcal{D})]].
\end{equation}
Here, the first term is the posterior predictive entropy, which can be computed analytically when a Gaussian process prior is used. The entropy in the second term, conditioned on the optimum, can be efficiently approximated using the expectation propagation method \cite{minka2013expectation}.

The core idea of entropy search is to enhance sampling efficiency by maximizing information gain about the optimum. This raises a natural question: how does more informative data compare to less informative data in terms of model discovery? In sections below, we will visualize and quantify such discrepancies using FIM-based metrics.

\section{Information Analysis Based on the Fisher Information Matrix}\label{sec:fim_analysis}

In this section, we develop a rigorous framework for quantifying the informativeness of trajectory data within the context of sparse identification of nonlinear dynamics (SINDy). We establish the Fisher Information Matrix (FIM) as a principled metric for assessing data quality and demonstrate how its spectral properties directly relate to parameter estimation accuracy. Furthermore, we provide theoretical justification for the efficacy of statistical bagging methods through the lens of information geometry.

\subsection{Problem Formulation and the Fisher Information Matrix}

Consider a dynamical system governed by the ordinary differential equation
\begin{equation}\label{eq:ode_system}
    \dot{\mathbf{x}} = \mathbf{f}(\mathbf{x}; \boldsymbol{\theta}),
\end{equation}
where $\mathbf{x} \in \mathbb{R}^n$ denotes the state vector and $\boldsymbol{\theta} \in \mathbb{R}^p$ represents the unknown parameters to be identified. Within the SINDy framework, we seek a sparse representation
\begin{equation}\label{eq:sindy_form}
    \dot{\mathbf{X}} = \boldsymbol{\Theta}(\mathbf{X}) \boldsymbol{\Xi},
\end{equation}
where $\mathbf{X} \in \mathbb{R}^{m \times n}$ is the data matrix of $m$ temporal measurements, $\boldsymbol{\Theta}(\mathbf{X}) \in \mathbb{R}^{m \times q}$ is the library matrix of $q$ candidate nonlinear functions evaluated on the data, and $\boldsymbol{\Xi} \in \mathbb{R}^{q \times n}$ is the sparse coefficient matrix encoding the governing dynamics.

The SINDy optimization problem can be formulated as
\begin{equation}\label{eq:sindy_optimization}
    \hat{\boldsymbol{\Xi}} = \argmin_{\boldsymbol{\Xi}} \frac{1}{2} \|\boldsymbol{\Theta}(\mathbf{X})\boldsymbol{\Xi} - \dot{\mathbf{X}}\|_F^2 + \mathcal{R}(\boldsymbol{\Xi}),
\end{equation}
where $\|\cdot\|_F$ denotes the Frobenius norm and $\mathcal{R}(\boldsymbol{\Xi})$ is a sparsity-promoting regularizer. For a single state variable, this reduces to the linear regression model
\begin{equation}\label{eq:linear_model}
    \mathbf{y} = \mathbf{A}\boldsymbol{\xi} + \boldsymbol{\varepsilon},
\end{equation}
where $\mathbf{y} = \dot{\mathbf{x}} \in \mathbb{R}^m$ is the time derivative, $\mathbf{A} = \boldsymbol{\Theta}(\mathbf{X}) \in \mathbb{R}^{m \times q}$ is the library matrix, $\boldsymbol{\xi} \in \mathbb{R}^q$ is the coefficient vector, and $\boldsymbol{\varepsilon} \sim \mathcal{N}(\mathbf{0}, \sigma^2 \mathbf{I}_m)$ represents measurement noise.

Under the Gaussian noise assumption, the log-likelihood function for observing $\mathbf{y}$ given parameters $\boldsymbol{\xi}$ is
\begin{equation}\label{eq:log_likelihood}
    \ell(\boldsymbol{\xi}; \mathbf{y}, \mathbf{A}) = -\frac{1}{2\sigma^2}\|\mathbf{A}\boldsymbol{\xi} - \mathbf{y}\|_2^2 - \frac{m}{2}\log(2\pi\sigma^2).
\end{equation}

The Fisher Information Matrix (FIM) quantifies the amount of information that the observed data carries about the unknown parameters. For the model in Equation~\eqref{eq:linear_model}, the FIM is defined as the negative expected Hessian of the log-likelihood:
\begin{equation}\label{eq:fim_definition}
    \mathbf{I}(\boldsymbol{\xi}) = -\mathbb{E}\left[\nabla_{\boldsymbol{\xi}}^2 \ell(\boldsymbol{\xi}; \mathbf{y}, \mathbf{A})\right] = -\mathbb{E}\left[\frac{\partial^2 \ell}{\partial \xi_i \partial \xi_j}\right]_{i,j=1}^{q}.
\end{equation}

\begin{proposition}[Fisher Information Matrix for Linear Regression]\label{prop:fim_linear}
For the linear model $\mathbf{y} = \mathbf{A}\boldsymbol{\xi} + \boldsymbol{\varepsilon}$ with $\boldsymbol{\varepsilon} \sim \mathcal{N}(\mathbf{0}, \sigma^2\mathbf{I})$, the Fisher Information Matrix is given by
\begin{equation}\label{eq:fim_formula}
    \mathbf{I}(\boldsymbol{\xi}) = \frac{1}{\sigma^2}\mathbf{A}^\top\mathbf{A}.
\end{equation}
\end{proposition}

\begin{proof}
The gradient of the log-likelihood with respect to $\boldsymbol{\xi}$ is
\begin{equation}
    \nabla_{\boldsymbol{\xi}} \ell = \frac{1}{\sigma^2}\mathbf{A}^\top(\mathbf{y} - \mathbf{A}\boldsymbol{\xi}).
\end{equation}
The Hessian is therefore
\begin{equation}
    \nabla_{\boldsymbol{\xi}}^2 \ell = -\frac{1}{\sigma^2}\mathbf{A}^\top\mathbf{A},
\end{equation}
which is constant with respect to $\mathbf{y}$. Thus, $\mathbf{I}(\boldsymbol{\xi}) = -\nabla_{\boldsymbol{\xi}}^2 \ell = \frac{1}{\sigma^2}\mathbf{A}^\top\mathbf{A}$.
\end{proof}

\begin{remark}
The FIM $\mathbf{I}(\boldsymbol{\xi})$ is independent of the true parameter value $\boldsymbol{\xi}$, which is a characteristic property of linear Gaussian models. This independence simplifies the analysis considerably, as the informativeness of data can be assessed without knowledge of the ground truth parameters.
\end{remark}

\subsection{Spectral Interpretation and Information Metrics}

The eigenvalue decomposition of the FIM provides geometric insight into the parameter estimation landscape. Let
\begin{equation}\label{eq:fim_eigendecomp}
    \mathbf{I}(\boldsymbol{\xi}) = \mathbf{V}\boldsymbol{\Lambda}\mathbf{V}^\top = \sum_{k=1}^{q} \lambda_k \mathbf{v}_k \mathbf{v}_k^\top,
\end{equation}
where $\lambda_1 \geq \lambda_2 \geq \cdots \geq \lambda_q \geq 0$ are the eigenvalues and $\mathbf{v}_1, \ldots, \mathbf{v}_q$ are the corresponding orthonormal eigenvectors.

The eigenvalues admit a direct interpretation via the singular value decomposition of $\mathbf{A}$. If $\mathbf{A} = \mathbf{U}\boldsymbol{\Sigma}\mathbf{V}^\top$ with singular values $s_1 \geq s_2 \geq \cdots \geq s_{\min(m,q)} \geq 0$, then
\begin{equation}\label{eq:eigenvalue_singular}
    \lambda_k = \frac{s_k^2}{\sigma^2}, \quad k = 1, \ldots, q.
\end{equation}

\begin{definition}[Directional Information]
The information content along a direction $u \in \mathbb{R}^q$ with $\|u\|_2 = 1$ is defined via the Rayleigh quotient~\cite{horn2012matrix}
\begin{equation}\label{eq:directional_info}
    \mathcal{I}_{\mathbf{u}} = \mathbf{u}^\top \mathbf{I}(\boldsymbol{\xi}) \mathbf{u} = \frac{1}{\sigma^2}\|\mathbf{A}\mathbf{u}\|_2^2.
\end{equation}
\end{definition}

This quantity measures how sensitively the likelihood function responds to perturbations of the parameter vector along direction $\mathbf{u}$. A large value indicates that the data strongly constrains parameter variations in that direction, while a small value suggests parameter uncertainty.

\begin{proposition}[Extremal Information Directions]\label{prop:extremal_directions}
The maximum and minimum directional information are achieved along the principal eigenvectors:
\begin{align}
    \max_{\|\mathbf{u}\|=1} \mathcal{I}_{\mathbf{u}} &= \lambda_1 = \lambda_{\max}(\mathbf{I}), \quad \text{achieved at } \mathbf{u} = \mathbf{v}_1, \label{eq:max_info}\\
    \min_{\|\mathbf{u}\|=1} \mathcal{I}_{\mathbf{u}} &= \lambda_q = \lambda_{\min}(\mathbf{I}), \quad \text{achieved at } \mathbf{u} = \mathbf{v}_q. \label{eq:min_info}
\end{align}
\end{proposition}

This result follows directly from the Rayleigh quotient characterization of eigenvalues. The eigenvector $\mathbf{v}_1$ corresponds to the most \emph{informative} direction in parameter space - the direction along which the data most strongly constrains the parameters - while $\mathbf{v}_q$ corresponds to the least informative direction.

To provide a scalar summary of data informativeness, we propose the following metrics derived from the FIM spectrum:

\begin{definition}[FIM-Based Information Metrics]\label{def:fim_metrics}
Given a Fisher Information Matrix $\mathbf{I}$ with eigenvalues $\lambda_1 \geq \cdots \geq \lambda_q$, we define:
\begin{enumerate}[(i)]
    \item \textbf{Maximum eigenvalue metric:} $\mathcal{M}_{\max} = \lambda_1$, capturing the strongest constraint direction.
    
    \item \textbf{Trace metric (A-optimality):} $\mathcal{M}_{\mathrm{tr}} = \mathrm{tr}(\mathbf{I}) = \sum_{k=1}^q \lambda_k$, measuring total information.
    
    \item \textbf{Determinant metric (D-optimality):} $\mathcal{M}_{\det} = \det(\mathbf{I}) = \prod_{k=1}^q \lambda_k$, quantifying the volume of the confidence ellipsoid.
    
    \item \textbf{Minimum eigenvalue metric (E-optimality):} $\mathcal{M}_{\min} = \lambda_q$, capturing the weakest constraint direction.
    
    \item \textbf{Condition number:} $\kappa(\mathbf{I}) = \lambda_1/\lambda_q$, assessing the spectral spread of information across directions.
    
    \item \textbf{Spectral skewness:} 
    \begin{equation}\label{eq:spectral_skewness}
        \mathcal{S} = \frac{\frac{1}{q}\sum_{k=1}^q (\lambda_k - \bar{\lambda})^3}{\left(\frac{1}{q}\sum_{k=1}^q (\lambda_k - \bar{\lambda})^2\right)^{3/2}},
    \end{equation}
    where $\bar{\lambda} = \frac{1}{q}\sum_k \lambda_k$, characterizing the asymmetry of information distribution.
    
    \item \textbf{Effective rank:}
    \begin{equation}\label{eq:effective_rank}
        r_{\mathrm{eff}}(\mathbf{I}) = \exp\left(-\sum_{k=1}^q \tilde{\lambda}_k \log \tilde{\lambda}_k\right), \quad \tilde{\lambda}_k = \frac{\lambda_k}{\sum_{j=1}^q \lambda_j},
    \end{equation}
    measuring the effective number of non-negligible eigenvalues via the exponential of Shannon entropy.
\end{enumerate}
\end{definition}

Table~\ref{tab:scalar_metrics} summarizes these metrics and their interpretations in the context of experimental design and model identification.

\begin{table}[htbp]
\centering
\renewcommand{\arraystretch}{1.3}
\begin{tabular}{@{}>{\centering\arraybackslash}p{3.2cm}>{\centering\arraybackslash}p{3.5cm}>{\centering\arraybackslash}p{4.0cm}>{\centering\arraybackslash}p{2.8cm}@{}}
\toprule
\textbf{Metric} & \textbf{Formula} & \textbf{Interpretation} & \textbf{Optimization Goal} \\
\midrule
Trace (A-opt) & $\mathrm{tr}(\mathbf{I})$ & Total information & Maximize \\
Determinant (D-opt) & $\det(\mathbf{I})$ & Confidence volume$^{-1}$ & Maximize \\
Max.\ eigenvalue & $\lambda_1 = \lambda_{\max}(\mathbf{I})$ & Strongest constraint direction & Maximize \\
Min.\ eigenvalue (E-opt) & $\lambda_{\min}(\mathbf{I})$ & Weakest constraint direction & Maximize \\
Condition number & $\kappa(\mathbf{I}) = \lambda_1/\lambda_q$ & Spectral spread & Minimize \\
Effective rank & $r_{\mathrm{eff}}(\mathbf{I})$ & Constrained dimensions & Context-dependent$^{\dagger}$ \\
Spectral skewness & $\mathcal{S}$ & Eigenvalue asymmetry & Context-dependent$^{\dagger}$ \\
\bottomrule
\end{tabular}
\\ [2pt]
{\small $^{\dagger}$Meaningful only when combined with other metrics such as $\lambda_1$ or $\mathrm{tr}(\mathbf{I})$.}
\caption{Summary of FIM metrics for assessing data informativeness in SINDy}
\label{tab:scalar_metrics}
\end{table}

\subsection{Consequences of ill-conditioned FIM for SINDy}\label{subsec:fim_pathologies}

The spectral properties of the Fisher Information Matrix have direct and consequential implications for the performance of SINDy. When the FIM is poorly conditioned (large $\kappa(\mathbf{I})$) or spectrally sparse (small $r_{\mathrm{eff}}(\mathbf{I})$), several pathological behaviors emerge that compromise model discovery.

\subsubsection{Parameter Non-Identifiability}

When one or more eigenvalues $\lambda_k \approx 0$, the corresponding parameter directions $\mathbf{v}_k$ become practically non-identifiable: the data provides negligible information to distinguish among parameter values along these directions. Formally, the Cram\'er-Rao bound implies
\begin{equation}\label{eq:variance_blowup}
    \mathrm{Var}(\hat{\xi}_k) \geq \lambda_k^{-1} \to \infty \quad \text{as} \quad \lambda_k \to 0,
\end{equation}
where $\hat{\xi}_k = \mathbf{v}_k^\top \hat{\boldsymbol{\xi}}$ is the estimated coefficient projected onto direction $\mathbf{v}_k$.

In the SINDy context, this manifests as coefficient ambiguity, where multiple coefficient vectors $\boldsymbol{\xi}$ yield nearly identical fits to the data since variations along the null directions of $\mathbf{I}$ are undetectable. Additionally, the solution becomes dominated by the regularizer $\mathcal{R}(\boldsymbol{\xi})$ rather than the data, making results sensitive to hyperparameter choices - a phenomenon we term regularization dependence.

\subsubsection{Spurious Term Selection}

A spectrally sparse FIM creates conditions favorable for false discoveries - the erroneous inclusion of inactive library terms or exclusion of active terms. Consider two library terms with true coefficients $\xi_i \neq 0$ (active) and $\xi_j = 0$ (inactive). If the directions distinguishing these terms lie in the low-eigenvalue subspace of $\mathbf{I}$, then noise perturbations can cause $|\hat{\xi}_j| > |\hat{\xi}_i|$, leading to incorrect sparsity patterns. Moreover, the sequential thresholding procedure in SINDy may prematurely eliminate the true term $\xi_i$ while retaining the spurious term $\xi_j$.

The probability of such misidentification scales with the condition number:
\begin{equation}\label{eq:misid_prob}
    \mathbb{P}(\text{spurious selection}) \propto \kappa(\mathbf{I}) \cdot \frac{\sigma}{\min_{k:\xi_k \neq 0}|\xi_k|},
\end{equation}
where the ratio $\sigma/\min_k|\xi_k|$ represents the signal-to-noise ratio for the weakest active term.

\subsubsection{Numerical Instability}

The least-squares solution underlying SINDy involves computing
\begin{equation}\label{eq:normal_equations}
    \hat{\boldsymbol{\xi}} = (\mathbf{A}^\top\mathbf{A})^{-1}\mathbf{A}^\top\mathbf{y} = \sigma^2 \mathbf{I}^{-1}\mathbf{A}^\top\mathbf{y}.
\end{equation}
When $\mathbf{I}$ is ill-conditioned, the matrix inversion amplifies numerical errors:
\begin{equation}\label{eq:numerical_error}
    \frac{\|\delta\hat{\boldsymbol{\xi}}\|}{\|\hat{\boldsymbol{\xi}}\|} \leq \kappa(\mathbf{I}) \cdot \frac{\|\delta\mathbf{y}\|}{\|\mathbf{y}\|},
\end{equation}
where $\delta\mathbf{y}$ represents measurement noise or numerical round-off. For $\kappa(\mathbf{I}) \gg 1$, even small data perturbations produce large coefficient variations, undermining reproducibility.

\subsubsection{Multicollinearity in the Library}

High condition numbers often indicate multicollinearity - near-linear dependencies among library columns. In dynamical systems, this arises naturally when the trajectory remains confined to a low-dimensional submanifold of state space, when certain library terms (e.g., $x^2$ and $xy$) are highly correlated along the observed trajectory, or when the sampling rate is mismatched to the system timescales.

Multicollinearity causes the FIM to have a sloppy spectrum \cite{transtrum2015perspective}: a hierarchy of eigenvalues spanning many orders of magnitude. The ``stiff'' directions (large $\lambda_k$) correspond to well-constrained parameter combinations, while ``sloppy'' directions (small $\lambda_k$) represent parameter trade-offs that leave the fit essentially unchanged.

\subsubsection{Remediation Strategies}

The pathologies above motivate several remediation strategies, each interpretable through the FIM framework. Discriminative sampling - selecting trajectory segments or initial conditions that maximize $\lambda_{\min}(\mathbf{I})$ or minimize $\kappa(\mathbf{I})$ - is developed in subsequent sections as a primary focus of this work. Library pruning removes highly correlated library terms to reduce multicollinearity, improving the condition number at the cost of model expressiveness. The sequential thresholded least-squares (STLS) regularization employed in SINDy iteratively removes small coefficients, effectively projecting onto well-constrained subspaces; this sparsity-promoting approach implicitly avoids the poorly conditioned directions by eliminating terms whose coefficients fall below the threshold. Finally, ensemble methods such as bootstrap aggregation diversify the spectral support by averaging over multiple resampled datasets, as analyzed in Section~\ref{subsec:bagging}.

\begin{proposition}[Connection to Estimator Variance]\label{prop:cramer_rao}
Under standard regularity conditions, the Cram\'er-Rao lower bound establishes that any unbiased estimator $\hat{\boldsymbol{\xi}}$ satisfies
\begin{equation}\label{eq:cramer_rao}
    \mathrm{Cov}(\hat{\boldsymbol{\xi}}) \succeq \mathbf{I}(\boldsymbol{\xi})^{-1},
\end{equation}
where $\mathbf{A} \succeq \mathbf{B}$ denotes that $\mathbf{A} - \mathbf{B}$ is positive semidefinite. Consequently:
\begin{enumerate}[(i)]
    \item The variance along eigenvector $\mathbf{v}_k$ is bounded below by $\lambda_k^{-1}$.
    \item Small eigenvalues correspond to directions of high parameter uncertainty.
    \item If $\lambda_q = 0$, the parameters are not identifiable from the data.
\end{enumerate}
\end{proposition}

\subsection{Comparison with Entropy-Based Metrics}\label{subsec:entropy_comparison}

Before proceeding to multi-trajectory analysis, we contextualize the FIM approach relative to entropy-based metrics commonly employed in time series analysis (see Section~\ref{sec:preliminaries} for detailed definitions). Table~\ref{tab:metric_comparison} summarizes the key distinctions.

\begin{table}[htbp]
\centering
\renewcommand{\arraystretch}{1.3}
\begin{tabular}{@{}>{\centering\arraybackslash}p{2.8cm}>{\centering\arraybackslash}p{4.2cm}>{\centering\arraybackslash}p{4.2cm}>{\centering\arraybackslash}p{3.5cm}@{}}
\toprule
\textbf{Metric} & \textbf{Measures} & \textbf{Advantages} & \textbf{Limitations} \\
\midrule
Fisher Information Matrix & Information content for parameter estimation & Direct connection to estimator variance; parameter-specific; no hyperparameters & Requires model specification; assumes Gaussian noise \\
Approximate Entropy (ApEn) & Regularity and pattern recurrence in time series & Model-free; captures temporal structure & Sensitive to template length $m$ and threshold $r$; self-comparison bias \\
Sample Entropy (SampEn) & Conditional probability of pattern continuation & Unbiased relative to ApEn; length-independent & Still requires $m$ and $r$; limited to univariate series \\
\bottomrule
\end{tabular}
\caption{Comparison of information metrics for dynamical system identification}
\label{tab:metric_comparison}
\end{table}

The Approximate Entropy (ApEn) \cite{pincus1991approximate} and Sample Entropy (SampEn) \cite{richman2000physiological} quantify the regularity of time series data through pattern matching, as introduced in Sections~\ref{subsec:apen} and~\ref{subsec:sampen}. For a sequence $\{u_1, \ldots, u_N\}$, template length $m$, and tolerance $r$, these metrics assess the likelihood that similar patterns of length $m$ remain similar at length $m+1$. Specifically, defining the correlation sums
\begin{align}
    C_i^m(r) &= \frac{1}{N-m+1}\sum_{j=1}^{N-m+1} \mathbf{1}\left[\|\mathbf{u}_m(i) - \mathbf{u}_m(j)\|_\infty \leq r\right], \label{eq:correlation_apen}\\
    B_i^m(r) &= \frac{1}{N-m-1}\sum_{\substack{j=1 \\ j \neq i}}^{N-m} \mathbf{1}\left[\|\mathbf{u}_m(i) - \mathbf{u}_m(j)\|_\infty \leq r\right], \label{eq:correlation_sampen}
\end{align}
where $\mathbf{u}_m(i) = (u_i, u_{i+1}, \ldots, u_{i+m-1})^\top$, the entropies are computed as
\begin{align}
    \mathrm{ApEn}(m, r, N) &= \phi^m(r) - \phi^{m+1}(r), \quad \phi^m(r) = \frac{1}{N-m+1}\sum_{i=1}^{N-m+1}\log C_i^m(r), \label{eq:apen_def}\\
    \mathrm{SampEn}(m, r, N) &= -\log\frac{A^m(r)}{B^m(r)}, \label{eq:sampen_def}
\end{align}
where $A^m(r)$ and $B^m(r)$ are averages of the respective correlation sums at lengths $m+1$ and $m$. As discussed in Section~\ref{subsec:sampen}, the key distinction is that ApEn includes self-comparisons in its probability estimates, introducing a bias that becomes significant when few matching patterns exist, whereas SampEn excludes self-matches and is therefore statistically more consistent \cite{richman2000physiological}.

While entropy metrics provide valuable insights into time series complexity without requiring a parametric model, they suffer from several limitations in the context of model discovery. First, both ApEn and SampEn require specification of the embedding dimension $m$ and tolerance $r$, with results often highly sensitive to these choices; the optimal values may vary across different temporal segments of a chaotic trajectory. Second, entropy metrics characterize intrinsic time series complexity but do not directly quantify how well the data constrains model parameters - a trajectory segment may exhibit high regularity (low entropy) yet provide poor information for distinguishing among candidate model terms. Third, standard entropy metrics are designed for scalar time series, whereas SINDy operates on multivariate state vectors and their nonlinear transformations.

In contrast, the FIM directly quantifies parameter identifiability through the spectral properties of $\mathbf{A}^\top\mathbf{A}$, providing actionable guidance for discriminative sampling without hyperparameter tuning. As established in Section~\ref{subsec:fim_prelim}, the FIM captures the sensitivity of the likelihood function to parameter perturbations, yielding a direct connection to estimator variance via the Cram\'er-Rao bound. The FIM also naturally extends to multidimensional systems through block-diagonal structures, as we develop in the following subsection.

\subsection{Multiple Trajectories and Information Additivity}\label{subsec:multi_trajectory}

We now analyze how multiple trajectories contribute to the overall informativeness of training data, providing theoretical justification for ensemble methods in SINDy.

Consider $K$ independent trajectories, each yielding data matrices $\mathbf{A}_1, \ldots, \mathbf{A}_K \in \mathbb{R}^{m_k \times q}$, where $m_k$ is the number of temporal samples in trajectory $k$. The combined data matrix is
\begin{equation}\label{eq:combined_data}
    \bar{\mathbf{A}} = \begin{bmatrix} \mathbf{A}_1 \\ \mathbf{A}_2 \\ \vdots \\ \mathbf{A}_K \end{bmatrix} \in \mathbb{R}^{M \times q}, \quad M = \sum_{k=1}^K m_k.
\end{equation}

\begin{theorem}[Additivity of Fisher Information]\label{thm:fim_additivity}
For independent trajectories with common noise variance $\sigma^2$, the aggregate Fisher Information Matrix satisfies
\begin{equation}\label{eq:fim_sum}
    \mathbf{I}_{\Sigma} = \frac{1}{\sigma^2}\bar{\mathbf{A}}^\top\bar{\mathbf{A}} = \frac{1}{\sigma^2}\sum_{k=1}^K \mathbf{A}_k^\top\mathbf{A}_k = \sum_{k=1}^K \mathbf{I}_k,
\end{equation}
where $\mathbf{I}_k = \frac{1}{\sigma^2}\mathbf{A}_k^\top\mathbf{A}_k$ is the FIM contribution from trajectory $k$.
\end{theorem}

\begin{proof}
By direct computation,
\begin{equation}
    \bar{\mathbf{A}}^\top\bar{\mathbf{A}} = \begin{bmatrix} \mathbf{A}_1^\top & \cdots & \mathbf{A}_K^\top \end{bmatrix} \begin{bmatrix} \mathbf{A}_1 \\ \vdots \\ \mathbf{A}_K \end{bmatrix} = \sum_{k=1}^K \mathbf{A}_k^\top\mathbf{A}_k.
\end{equation}
\end{proof}

This additivity property has profound implications for the spectral structure of the combined FIM. We now establish key results relating the eigenvalues of individual and combined information matrices.

\begin{theorem}[Eigenvalue Bounds for Aggregate FIM]\label{thm:eigenvalue_bounds}
Let $\lambda_j(\mathbf{I}_k)$ denote the $j$-th largest eigenvalue of $\mathbf{I}_k$, and let $\lambda_j(\mathbf{I}_{\Sigma})$ denote the $j$-th largest eigenvalue of $\mathbf{I}_{\Sigma}$. Then:
\begin{enumerate}[(i)]
    \item \textbf{Monotonicity:} For all $j \in \{1, \ldots, q\}$,
    \begin{equation}\label{eq:monotonicity}
        \lambda_j(\mathbf{I}_{\Sigma}) \geq \max_{k \in \{1,\ldots,K\}} \lambda_j(\mathbf{I}_k).
    \end{equation}
    
    \item \textbf{Trace preservation:}
    \begin{equation}\label{eq:trace_preservation}
        \sum_{j=1}^q \lambda_j(\mathbf{I}_{\Sigma}) = \sum_{k=1}^K \sum_{j=1}^q \lambda_j(\mathbf{I}_k).
    \end{equation}
    
    \item \textbf{Maximum eigenvalue bound:}
    \begin{equation}\label{eq:max_eigenvalue_bound}
        \max_k \lambda_1(\mathbf{I}_k) \leq \lambda_1(\mathbf{I}_{\Sigma}) \leq \sum_{k=1}^K \lambda_1(\mathbf{I}_k).
    \end{equation}
\end{enumerate}
\end{theorem}

\begin{proof}
\textit{(i)} By the Courant-Fischer minimax theorem, for any Hermitian matrix $\mathbf{H}$,
\begin{equation}
    \lambda_j(\mathbf{H}) = \max_{\dim(S)=j} \min_{\mathbf{u} \in S, \|\mathbf{u}\|=1} \mathbf{u}^\top \mathbf{H} \mathbf{u}.
\end{equation}
Since $\mathbf{I}_{\Sigma} = \sum_k \mathbf{I}_k$ with each $\mathbf{I}_k \succeq \mathbf{0}$, we have $\mathbf{I}_{\Sigma} \succeq \mathbf{I}_k$ for all $k$, and thus $\lambda_j(\mathbf{I}_{\Sigma}) \geq \lambda_j(\mathbf{I}_k)$ for all $j$ and $k$.

\textit{(ii)} This follows from the linearity of the trace operator: $\mathrm{tr}(\mathbf{I}_{\Sigma}) = \mathrm{tr}\left(\sum_k \mathbf{I}_k\right) = \sum_k \mathrm{tr}(\mathbf{I}_k)$.

\textit{(iii)} The lower bound follows from part (i) with $j=1$. For the upper bound, note that $\lambda_1(\mathbf{I}_{\Sigma}) = \|\mathbf{I}_{\Sigma}\|_{\mathrm{op}} \leq \sum_k \|\mathbf{I}_k\|_{\mathrm{op}} = \sum_k \lambda_1(\mathbf{I}_k)$ by the triangle inequality for the operator norm.
\end{proof}

\begin{corollary}[Information Monotonicity]\label{cor:info_monotonicity}
Adding trajectories never decreases the informativeness of the aggregate dataset. Specifically, if $\mathbf{I}_{K}$ denotes the aggregate FIM from $K$ trajectories and $\mathbf{I}_{K+1}$ includes an additional trajectory, then
\begin{equation}
    \mathbf{I}_{K+1} \succeq \mathbf{I}_K,
\end{equation}
and consequently $\lambda_j(\mathbf{I}_{K+1}) \geq \lambda_j(\mathbf{I}_K)$ for all $j$.
\end{corollary}

\subsection{Spectral Diversification and the Efficacy of Bootstrap Aggregation}\label{subsec:bagging}

We now provide a theoretical explanation for the empirical success of bootstrap aggregation (bagging) methods \cite{breiman1996bagging} in improving SINDy model robustness, particularly in noise-limited and data-limited regimes.

\subsubsection{The Spectral Deficiency Problem}

In practice, trajectory data from a single initial condition often produces a Fisher Information Matrix with a highly skewed spectrum: a few large eigenvalues corresponding to well-constrained parameter directions, and many near-zero eigenvalues corresponding to poorly constrained directions. This spectral deficiency manifests in several pathologies. When $\lambda_q \approx 0$, the inverse covariance bound $\mathbf{I}^{-1}$ becomes ill-conditioned, leading to high variance in parameter estimates along the corresponding eigenvectors (near-singular covariance). In the presence of noise, the sparse regression may incorrectly select or reject library terms whose contributions lie primarily along poorly constrained directions (spurious term selection). Furthermore, small perturbations to the data can cause large changes in estimated coefficients when the FIM is nearly singular (coefficient instability).

\begin{definition}[Spectral Gap]\label{def:spectral_gap}
The spectral gap of the FIM is defined as
\begin{equation}\label{eq:spectral_gap}
    \Delta = \lambda_1 - \lambda_q,
\end{equation}
with a large spectral gap indicating highly anisotropic information content.
\end{definition}

\subsubsection{Bagging as Spectral Diversification}

Bootstrap aggregation constructs $B$ bootstrap samples $\{(\mathbf{A}^{(b)}, \mathbf{y}^{(b)})\}_{b=1}^B$ by resampling with replacement from the original dataset. Each bootstrap sample yields a FIM
\begin{equation}\label{eq:bootstrap_fim}
    \mathbf{I}^{(b)} = \frac{1}{\sigma^2}(\mathbf{A}^{(b)})^\top\mathbf{A}^{(b)}.
\end{equation}

\begin{proposition}[Diversification of Principal Directions]\label{prop:diversification}
Let $\mathbf{v}_1^{(b)}$ denote the leading eigenvector of the $b$-th bootstrap FIM. Under random resampling with replacement:
\begin{enumerate}[(i)]
    \item The bootstrap principal directions $\{\mathbf{v}_1^{(b)}\}_{b=1}^B$ exhibit variation around the population principal direction $\mathbf{v}_1$.
    
    \item The effective combined information matrix
    \begin{equation}\label{eq:effective_fim}
        \bar{\mathbf{I}}_{\mathrm{eff}} = \frac{1}{B}\sum_{b=1}^B \mathbf{I}^{(b)}
    \end{equation}
    has a less skewed spectrum than any individual $\mathbf{I}^{(b)}$ with high probability.
    
    \item The condition number $\kappa(\bar{\mathbf{I}}_{\mathrm{eff}}) \leq \kappa(\mathbf{I}^{(b)})$ for most bootstrap samples $b$.
\end{enumerate}
\end{proposition}

The intuition is as follows: each bootstrap sample emphasizes slightly different data points, leading to FIMs with rotated principal directions. When aggregated, these rotated contributions fill in the poorly constrained directions of the original FIM, effectively diversifying the spectral support.

\begin{theorem}[Bagging Reduces Effective Variance]\label{thm:bagging_variance}
Let $\hat{\boldsymbol{\xi}}^{(b)}$ be the SINDy coefficient estimate from bootstrap sample $b$, and let $\bar{\boldsymbol{\xi}} = \frac{1}{B}\sum_{b=1}^B \hat{\boldsymbol{\xi}}^{(b)}$ be the bagged estimator. Under mild regularity conditions:
\begin{equation}
\mathrm{Var}(\bar{\xi}) < \mathrm{Var}(\hat{\xi}^{(1)})
\end{equation}
provided $\mathrm{Cov}(\hat{\xi}^{(1)}, \hat{\xi}^{(2)}) < \mathrm{Var}(\hat{\boldsymbol{\xi}}^{(1)})$.
\end{theorem}

\begin{proof}[Proof sketch]
By the law of total variance,
\begin{align}
    \mathrm{Var}(\bar{\boldsymbol{\xi}}) &= \frac{1}{B^2}\sum_{b=1}^B \mathrm{Var}(\hat{\boldsymbol{\xi}}^{(b)}) + \frac{1}{B^2}\sum_{b \neq b'} \mathrm{Cov}(\hat{\boldsymbol{\xi}}^{(b)}, \hat{\boldsymbol{\xi}}^{(b')}) \\
    &= \frac{1}{B}\mathrm{Var}(\hat{\boldsymbol{\xi}}^{(1)}) + \frac{B-1}{B}\mathrm{Cov}(\hat{\boldsymbol{\xi}}^{(1)}, \hat{\boldsymbol{\xi}}^{(2)}),
\end{align}
where we used exchangeability of the bootstrap samples. The correlation $\rho = \mathrm{Corr}(\hat{\boldsymbol{\xi}}^{(1)}, \hat{\boldsymbol{\xi}}^{(2)}) < 1$ due to the randomness of resampling, yielding variance reduction.
\end{proof}

\subsubsection{Geometric Interpretation}

The spectral diversification effect admits an elegant geometric interpretation. Consider the confidence ellipsoid associated with the FIM:
\begin{equation}\label{eq:confidence_ellipsoid}
    \mathcal{E}(\mathbf{I}) = \left\{\boldsymbol{\xi} : (\boldsymbol{\xi} - \hat{\boldsymbol{\xi}})^\top \mathbf{I} (\boldsymbol{\xi} - \hat{\boldsymbol{\xi}}) \leq \chi^2_{q,\alpha}\right\},
\end{equation}
where $\chi^2_{q,\alpha}$ is the chi-squared critical value. The principal axes of this ellipsoid are aligned with the eigenvectors $\mathbf{v}_k$, with semi-axis lengths proportional to $\lambda_k^{-1/2}$.

A highly skewed FIM spectrum corresponds to a highly elongated ellipsoid - confident in some directions but uncertain in others. The bagging procedure effectively averages over ellipsoids with varying orientations, producing a more spherical aggregate that reduces maximum uncertainty.

\begin{definition}[Effective Dimensionality]\label{def:effective_dim}
The effective dimensionality of the FIM, which quantifies the number of well-constrained parameter directions, is defined as
\begin{equation}\label{eq:effective_dim}
    d_{\mathrm{eff}}(\mathbf{I}) = \frac{\left(\sum_{k=1}^q \lambda_k\right)^2}{\sum_{k=1}^q \lambda_k^2} = \frac{\mathrm{tr}(\mathbf{I})^2}{\|\mathbf{I}\|_F^2}.
\end{equation}
\end{definition}

This quantity, also known as the participation ratio \cite{thouless1974electrons}, ranges from $1$ (all information concentrated in one direction) to $q$ (uniform distribution). Effective bagging increases $d_{\mathrm{eff}}$, indicating improved identifiability across multiple parameter directions.

\begin{remark}
The effective dimensionality $d_{\text{eff}}$ (participation ratio) and effective rank $r_{\text{eff}}$ (entropy-based) both measure the number of significant eigenvalues but differ in sensitivity: $d_{\text{eff}}$ is more influenced by dominant eigenvalues, while $r_{\text{eff}}$ is more sensitive to the tail of the spectrum.
\end{remark}

\begin{proposition}[Bagging Increases Effective Dimensionality]\label{prop:bagging_eff_dim}
Under the conditions of Proposition~\ref{prop:diversification}, the effective FIM from bagging satisfies
\begin{equation}
    d_{\mathrm{eff}}(\bar{\mathbf{I}}_{\mathrm{eff}}) \geq \mathbb{E}[d_{\mathrm{eff}}(\mathbf{I}^{(b)})],
\end{equation}
with equality only when all bootstrap FIMs share identical eigenvector orientations.
\end{proposition}

\subsection{Practical Implications for Discriminative Sampling}\label{subsec:practical}

The theoretical framework developed above yields actionable guidelines for improving sampling efficiency in SINDy applications. For initial condition selection, one should choose initial conditions that maximize $\mathrm{tr}(\mathbf{I})$ or minimize the condition number $\kappa(\mathbf{I})$, thereby ensuring balanced information across all parameter directions. In trajectory segmentation, one partitions trajectories into temporal segments and computes the FIM for each segment, prioritizing segments with high $\mathcal{M}_{\max}$ or low spectral skewness $\mathcal{S}$ for inclusion in the training set. For ensemble design, when using bootstrap aggregation, one should select bootstrap samples that maximize spectral diversity, as measured by the angular spread of principal eigenvectors across samples. Finally, appropriate stopping criteria terminate data collection when $d_{\mathrm{eff}}(\mathbf{I})$ exceeds a threshold or when incremental trajectories produce diminishing returns in terms of improvement.

These principles provide a rigorous foundation for the empirical observations reported in subsequent sections, where we demonstrate that information-guided sampling substantially improves model identification accuracy relative to uniform random sampling strategies.

\section{Visualization and Active learning on chaotic systems}

\subsection{FIM pattern on a single trajectory}

\subsubsection{Fixed point and limit cycle convergence}

\begin{figure}[t]
\centering
\begin{overpic}[scale=0.5]{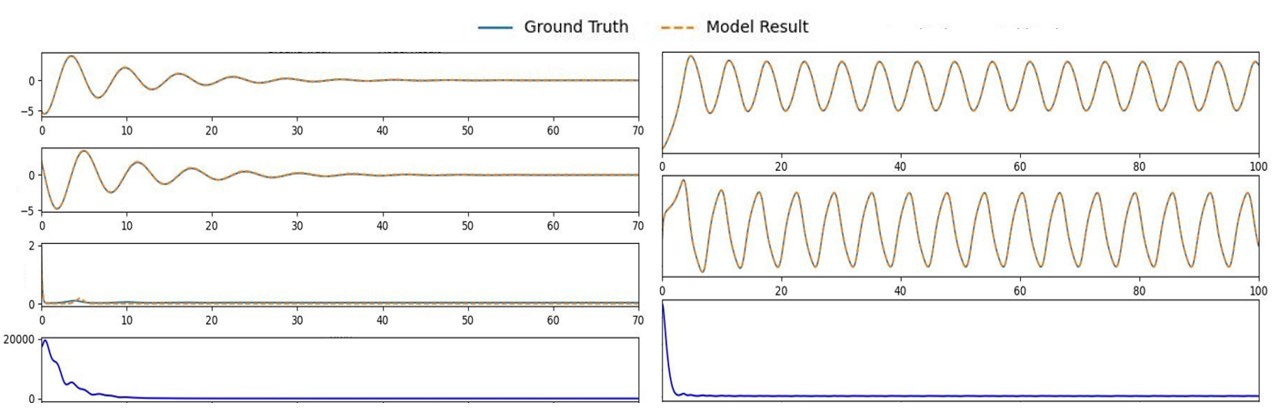}
    \put(-1.5,25.5){\footnotesize$\bf x(t)$}
    \put(-1.5,18){\footnotesize$\bf y(t)$}
    \put(-1.5,10.5){\footnotesize$\bf z(t)$}
    \put(-0.5,3){\footnotesize$\bf \lambda_1$}
    \put(48,-1.5){Time}
     \put(99,5){\footnotesize$\bf \lambda_1$}
     \put(99,24){\footnotesize$\bf x(t)$}
    \put(99,14){\footnotesize$\bf y(t)$}
    \put(-1.5,29){(a)}
    \put(99,29){(b)}
    \end{overpic}
\caption{Panel (a) illustrates the Rössler system converging to a centrally located fixed point under the parameter setting $a<0$, $b=0.2$, $c=5.7$. Panel (b) depicts the Van der Pol oscillator exhibiting a stable limit cycle for $\mu = 0.8$.
\label{fig:ros_vdp}}
\end{figure}
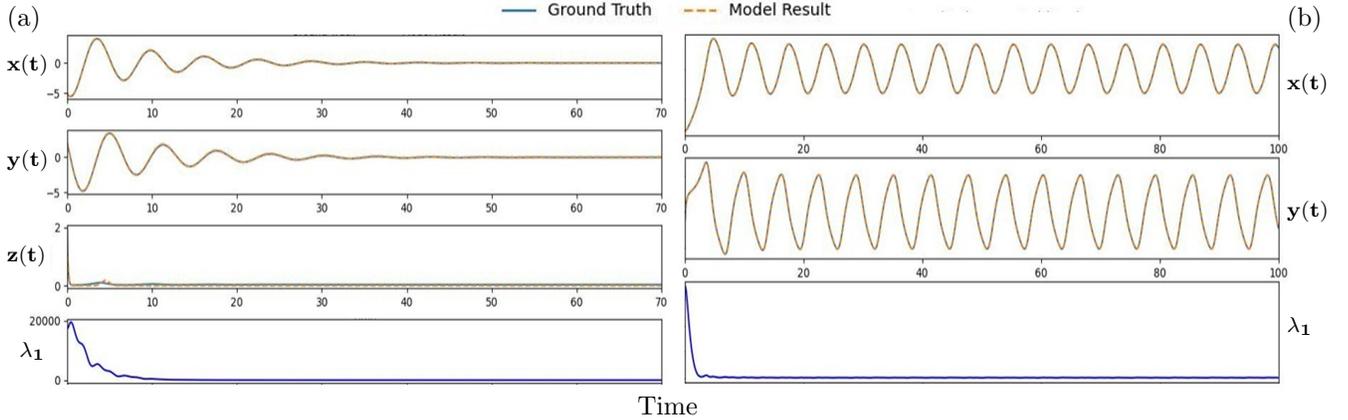

Below we discuss two cases of when the system converges to a fixed point or to a limit cycle, using a special case of the Rossler attractor and the Van der Pol oscillator respectively. The Rössler attractor, introduced by Otto Rössler in 1976, represents one of the simplest continuous-time dynamical systems capable of exhibiting low-dimensional chaos~\cite{rossler1976equation}. Formulated as three coupled nonlinear ODEs, the system exhibits linear growth, nonlinear twisting, and weak coupling that together produce its characteristic spiral chaotic attractor. Although the standard parameter regime yields sustained chaotic oscillations, the Rössler system also contains special parameter ranges in which trajectories converge to a stable fixed point or a simple periodic orbit, making it an important canonical example for illustrating transitions from steady-state behavior to complex chaos~\cite{sprott2003chaos}. Due to the structural simplicity and the clarity of its bifurcation structure, we will be using it as a testbed in both equilibrium state here and in chaotic settings later.

It is given by \cite{rossler1976equation}

\begin{align}
\dot{x_1} &= -x_2 - x_3 \nonumber \\
\dot{x_2} &= x_1 + ax_2 \\
\dot{x_3} &= b + x_3(x_1 - c) \nonumber
\end{align}

And the Van der Pol oscillator is given by 
\begin{equation}
\ddot x - \mu(1-x^2)\dot x + x = 0
\end{equation}

In Figure~\ref{fig:ros_vdp}, the largest eigenvalue of the Fisher information matrix is plotted against a sliding window of data blocks from input. From both cases, we observe a peak of information gain on the starting trajectory.

\subsubsection{Chaotic systems}

In contrast to the non-chaotic settings above, the FIM patterns for chaotic PDEs are more complicated. For the Lorenz system~\cite{sparrow2012lorenz} with two sets of initial
conditions, one closer to and one further away from the orbit, the training data is splitted into blocks with 20 measurements in each block. Then the largest eigenvalues of the FIM is plotted together with model simulations in Figure~\ref{fig:lor_1}. It's straightforward to observe that the pattern of the largest eigenvalues have variations between different sets of initial conditions. For initial settings with longer trajectories into the orbit, we would generally observe a spike of information gain at the start of the trajectory, while for other initial settings we would not observe much information gain at the beginning, but majorly near lobe switches instead.




\begin{figure}[t]
\centering
\begin{overpic}[scale=0.6]{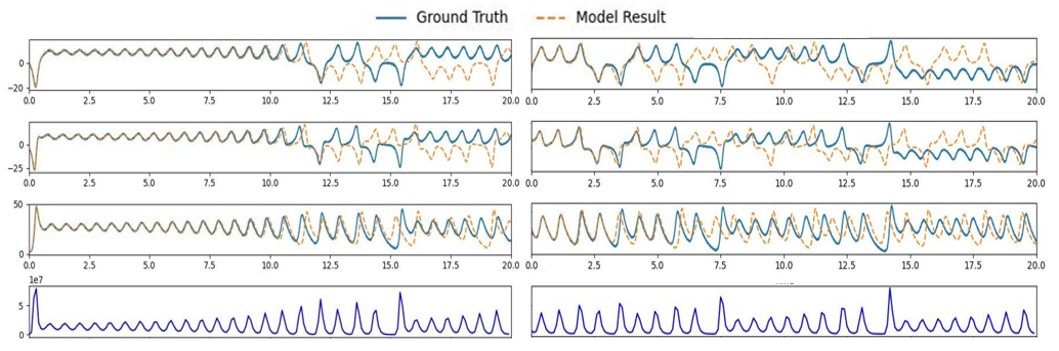}
    \put(-1.5,26){\footnotesize$\bf x(t)$}
    \put(-1.5,18.5){\footnotesize$\bf y(t)$}
    \put(-1.5,11){\footnotesize$\bf z(t)$}
    \put(-0.5,3.5){\footnotesize$\bf \lambda_1$}
    \put(48,-1.5){Time}
     \put(99,26){\footnotesize$\bf x(t)$}
    \put(99,18.5){\footnotesize$\bf y(t)$}
    \put(99,11){\footnotesize$\bf z(t)$}
    \put(99,3.5){\footnotesize$\bf\lambda_1$}
    \put(-1.5,29){(a)}
    \put(99,29){(b)}
    \end{overpic}
\caption{The largest eigenvalues $\lambda_1$ of the Fisher Information Matrix (FIM) are shown alongside model simulations for the chaotic Lorenz system. Panel (a) corresponds to a set of initial conditions farther from the orbit, while the panel (b) corresponds to initial condition closer to the orbit.
\label{fig:lor_1}}
\end{figure}
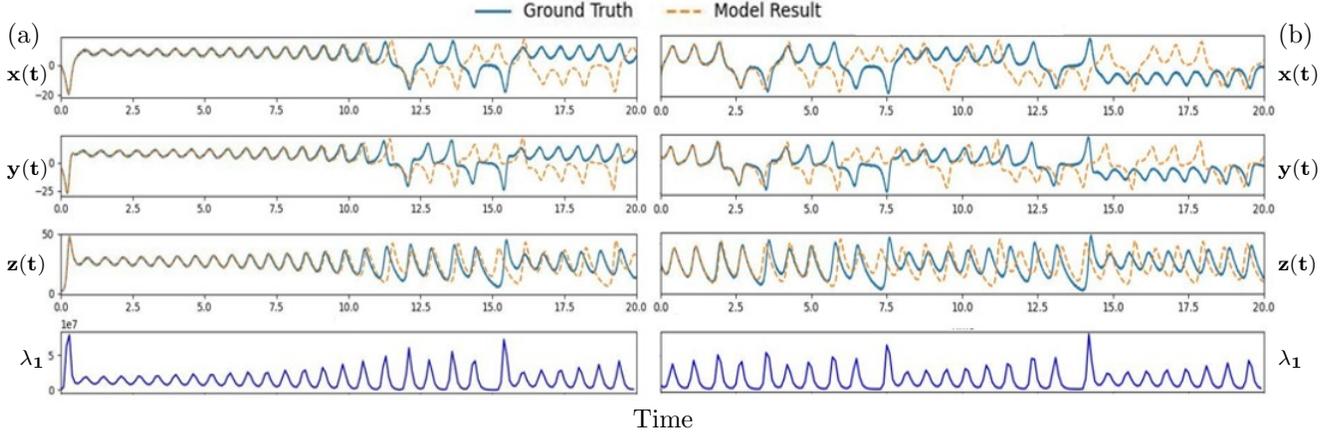

Switching to another chaotic system, we investigated how different initial conditions affect the amount of information gained from the initial observations. Using the chaotic Rössler attractor as a test case, we considered several representative initial-condition settings. When the initial condition lies close to the attractor’s orbit, the information gained from the starting data is minimal. In contrast, when the initial condition is placed further away from the orbit, the starting data becomes significantly more informative, leading to a marked increase in information gain. And apart from using solely the largest eigenvalue of the FIM as metric, we could also use skewness of its spectrum or a combination of both as metrics for information score, with similar observation that different initial conditions yield different information pattern.

\subsection{FIM pattern for multiple initial settings }

The drastic difference between initial conditions in terms of model efficiency as exhibited in the examples above naturally brings a new question - are all those disparities in learning outcomes due to informational differences? It calls for further investigation of the relationship between our info-metrics and model efficiency on a spatio-grid (and spatio-temporal grid for PDEs). In this section we interpret the coefficient loss as a mixture of two components - systematic loss, which occurs during the data collection process due to the formation of
the dynamical system itself, and informational loss, which occurs due to the trapping of
small magnitudes in the middle of the trajectories.

We plot on 2D grids of initial conditions as a visualization of the three losses - with different z-value choices of multiple dynamical systems.

\subsubsection{The Lorenz System}

When applying SINDy to the Lorenz system

\begin{align}
&\dot{x_1} = \sigma(x_2 - x_1) \nonumber \\
&\dot{x_2} = x_1(\rho - x_3) - x_2\\
&\dot{x_3} = x_1x_2 - \beta x_3, \nonumber
\end{align}

with hyperparameter setting $\sigma=10, \beta=2.66667, \rho=28$, we observe drastic differences in learnability from different initial settings. It's captured by comparing the $L_2$ coefficient loss learned from same SINDy model after acquiring the same amount of training data generated with different initials. The following plots provide an example of how SINDy-learned coefficients can distribute very differently with slight perturbations on the initials.

When different initials are plotted on a 2D-grid, an "X-shape" region that suffers high coefficient loss appears clearly near the origin. It reveals the less favorable starting points of a trajectory for a data-driven approach like SINDy. For the Lorenz system under current hyperparameter settings, such pattern is especially significant for z values near 9.

\begin{figure}[t]
\centering
\begin{overpic}[scale=0.52]{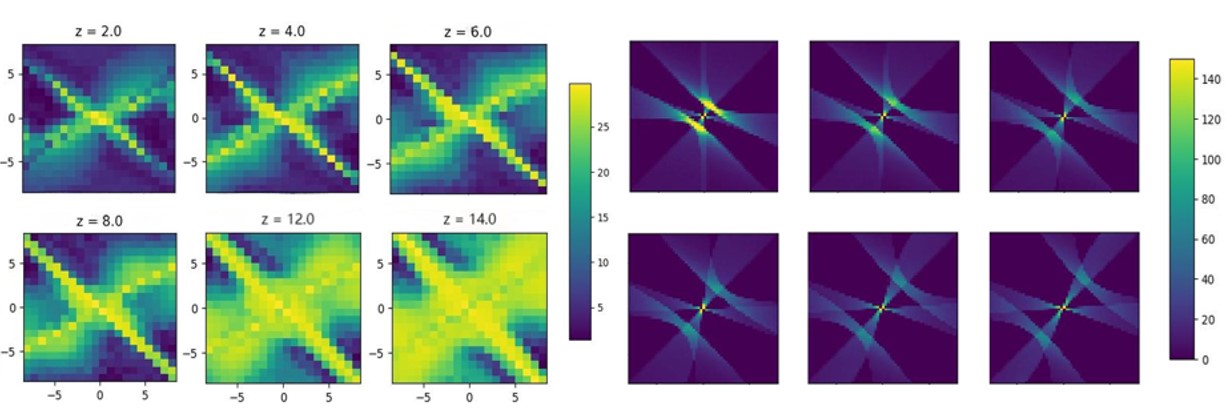}
    \put(-1.5,31){(a)}
    \put(48,31){(b)}
    \end{overpic}
\caption{The plots are defined on two-dimensional grids of initial conditions spanning different $(x, y)$ combinations, with distinct choices of the $z$-value across panels. Panel (a) presents the $L_2$ coefficient error (loss) of the Lorenz system obtained from training datasets of equal length, illustrating how reconstruction accuracy varies with the informativeness of the initial conditions. Panel (b) shows the number of extreme values encountered in training stage using data of the same length, serving as a proxy for local chaotic sensitivity and numerical stiffness. The $z$-value selections in panel (b) correspond to those used in the respective panels of (a), enabling a direct comparison between sampling discriminability, chaotic sensitivity, and model identification performance.
\label{fig:lor_2}}
\end{figure}

Our research reveals that the "X-shape" loss region is formed by two components - systematic loss and information loss. The systematic loss occurs during the data collection process due to the formation of the dynamical system itself. It impedes the learnability of SINDy model mainly with high computational bias induced by the surging of extreme values near the starting point of trajectories. It is less sensitive to perturbations on initial conditions. The informational loss, on the contrary, occurs due to the trapping of small magnitudes in the middle of the trajectories (which will be further explained in the following section), and is usually more sensitive to perturbations on the initials.

Firstly, on the systematic loss, the initial conditions near $y=x$ would lead to $\dot{x_1} \approx 0$ at the start of the trajectory, therefore causing computational bias on the training data. Such loss cannot be revealed by metrics from the FIM itself because it appears during the data collection process. A display on the extreme values of the derivative matrix (which is a component for the computation of the FIM) could provide a glimpse of how systematic loss plays a part in the impairment of SINDy. With different initials plotted on a 2D-grid, the right graph of Figure~\ref{fig:lor_2} shows the number of extreme values (either extremely small or large) in the derivative matrix at the starting region of the trajectories (within first 200 time steps). It coincide with the loss pattern near $y=x$ in the left graph of Figure~\ref{fig:lor_2}.

Secondly, the informational loss occurs for more complicated reasons within the trajectories. Initial conditions with $y \approx -x$ would lead to $\dot{x_1} \approx -2\sigma x_1$ and $\dot{x_2} = x_1(\rho + 1 - x_3)$. With $z \approx 9$ we would have $y \approx -x$ and $\dot{x_1} \approx -\dot{x_2}$ with small magnitudes at the beginning of the trajectory. It would stuck on $y=-x$ at small magnitude for a prolonged time. As a result, since y and x terms can be exchanged for low validation loss, it would falsify SINDy to include significantly more redundant library terms.

Such informational loss can be measured by our analysis in the prior sections, and our model uses metrics from the FIM to identify these less informational regions and guide the trajectories past them. With a fixed sampling interval, we could plot how much information is gained over the same time period under different initial conditions. And in Figure~\ref{fig:lor_3} we use a combination of the largest eigenvalue and the spectrum skewness of the FIM as the information metric. The lighter areas are initial conditions richer in information and darker areas are the contrary. The visualization coincides with the loss pattern near $y=-x$ in Figure~\ref{fig:lor_2}.

\begin{figure}[t]
\centering
\begin{overpic}[scale=0.57]{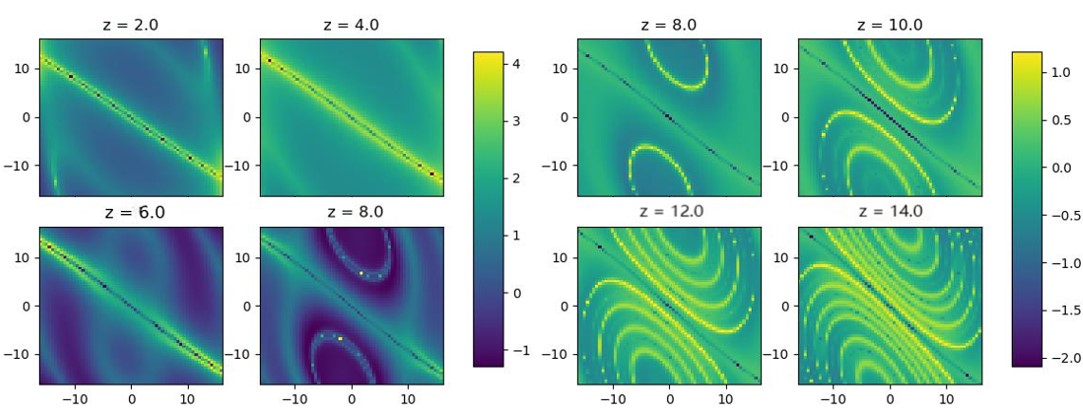}
    \put(-1.5,35){(a)}
    \put(48,35){(b)}
    \end{overpic}

\caption{The plots are defined on two-dimensional grids of initial conditions spanning different $(x, y)$ combinations, with distinct choices of the $z$-value across panels. The resulting information patterns of the Lorenz system under varying initial conditions are visualized using the largest eigenvalue and spectral skewness of the Fisher Information Matrix (FIM) as quantitative metrics. These measures characterize how informative each initial condition is for model identification, linking regions of higher information score to enhanced discriminative sampling capability and increased sensitivity to the underlying chaotic dynamics.
\label{fig:lor_3}}
\end{figure}

\subsubsection{The Rossler Attractor}

The Rossler attractor, which originally appeared from Otto Rossler's work on chemical kinetics is given by \cite{rossler1976equation}
\begin{align}
&\dot{x_1} = -x_2 - x_3 \nonumber \\
&\dot{x_2} = x_1 + ax_2 \\
&\dot{x_3} = b + x_3(x_1 - c) \nonumber
\end{align}

It exhibits chaotic behavior under hyperparameter setting $a = 0.2, b = 0.2, c = 5.7$. Similarly, when different initials are plotted on a 2D-grid, an "U-shape" region that suffers high coefficient loss appears clearly around the positive x-axis. It reveals the less favorable starting points of a trajectory for a data-driven approach like SINDy. In Figure~\ref{fig:ros_2}, similar pattern is exhibited across different choices of z-values for the initial.

\begin{figure}[t]
\centering
\begin{overpic}[scale=0.5]{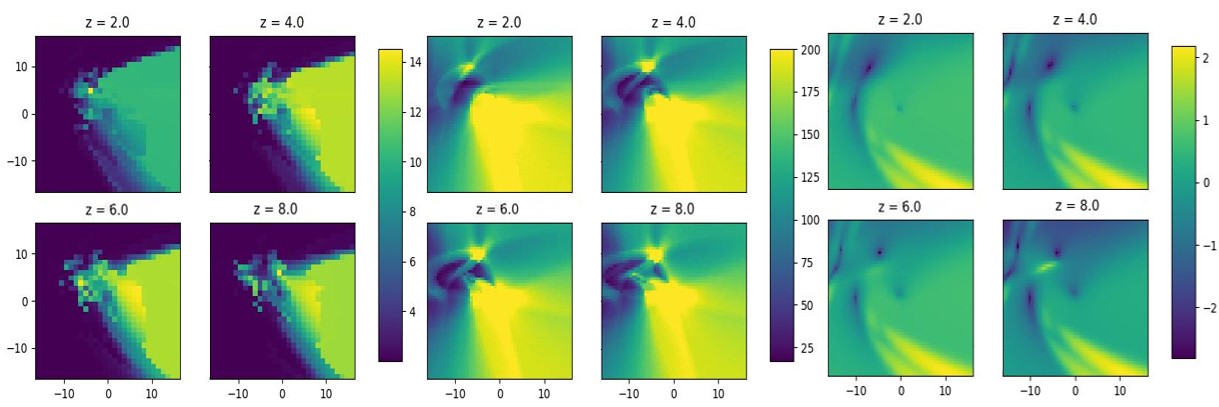}
    \put(-0.5,31.5){(a)}
    \put(32,31.5){(b)}
    \put(65,31.5){(c)}
    \end{overpic}

\caption{The plots are defined on two-dimensional grids of initial conditions spanning different $(x, y)$ combinations, with distinct choices of the $z$-value across panels. Panel (a) presents the $L_2$ coefficient error (loss) of the Rossler attractor obtained from training datasets of equal length; Panel (b) shows the number of extreme values encountered in training stage using data of the same length; Panel (c) presents information patterns of the Rossler attractor under varying initial conditions, using same metrics as in Fig~\ref{fig:lor_3}.
\label{fig:ros_2}}
\end{figure}

Similar to the Lorenz system above, the "U-shape" loss region is also formed by two components - systematic loss and information loss. Firstly, on the systematic loss, the initial conditions that lead to near-zero derivatives would cause computational bias on the starting trajectory of the training data. Such loss cannot be revealed by metrics from the FIM itself because it appears during the data collection process. A display on the extreme values of the derivative matrix (which is a component for the computation of the FIM) could provide a glimpse of how systematic loss plays a part in the impairment of SINDy. With different initials plotted on a 2D-grid, the middle panel of Figure~\ref{fig:ros_2} shows the number of extreme values (either extremely small or large) in the derivative matrix at the starting region of the trajectories (within first 200 time steps). It coincide majorly with the loss pattern in the left panel of Figure~\ref{fig:ros_2}.

However, there is one discrepancy in the coefficient loss plot that cannot be explained by the systematic loss plot alone, which is the good learnability in the lower region of initials on the x-y plane. Though the derivative matrix displays a surge in extreme values from the start of the trajectory, these small magnitude values are not trapped according to the plots of informational loss. Therefore they are not falsifying the SINDy algorithm to include redundant library terms. Similarly the informational loss can be plotted by our analysis in the prior sections. With a fixed sampling interval, we could plot how much information is gained over the same time period under different initial conditions. And in the right panel of Figure~\ref{fig:ros_2} we use the same information metric as in Figure~\ref{fig:lor_3}. The lighter areas are initial conditions richer in information and darker areas are the contrary. The visualization especially accounts for the good learnability of the lower region of initials on the x-y plane as in Figure~\ref{fig:ros_2}.

\begin{figure}[t]
\centering
\includegraphics[scale=0.5]{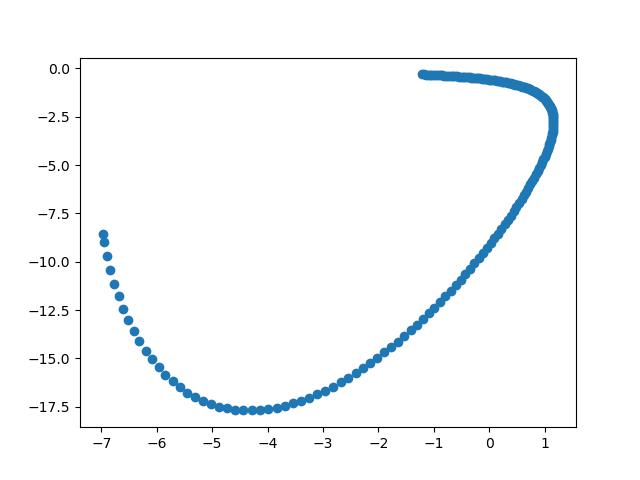}
\includegraphics[scale=0.5]{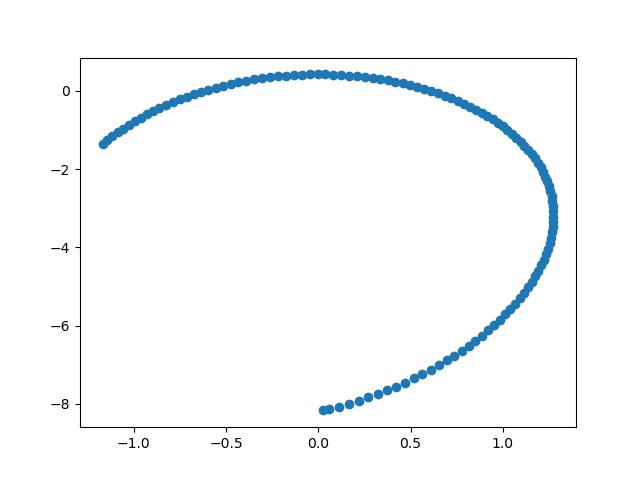}

\caption{A comparison between trapped and non-trapped trajectories on the x-z plane of derivatives. The scattered points are plotted on the x-z plane out of equally spaced time stepping.
\label{fig:trap1}}
\end{figure}

\begin{figure}[t]
\centering
\includegraphics[scale=0.28]{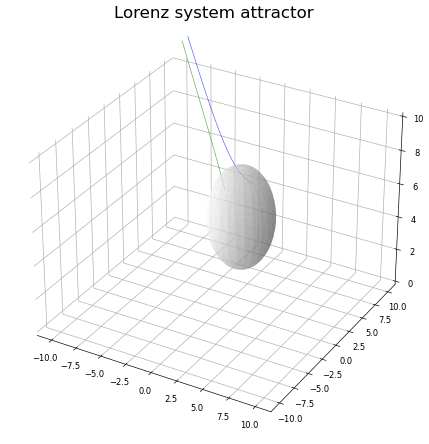}
\includegraphics[scale=0.28]{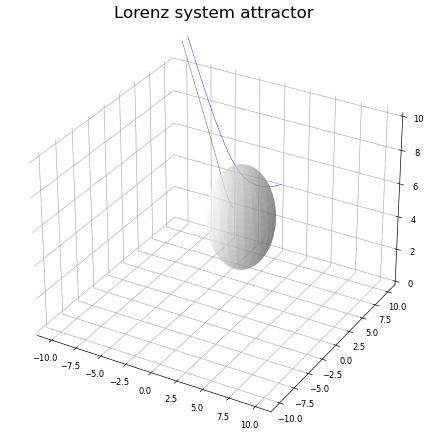}
\includegraphics[scale=0.28]{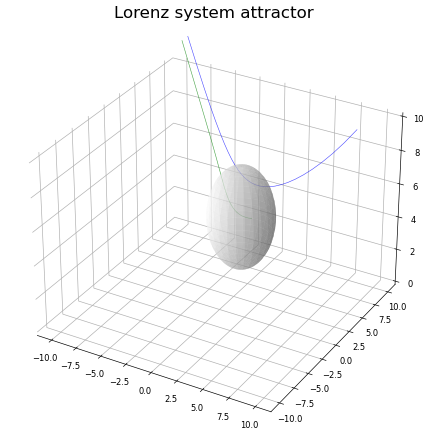}

\caption{A comparison between trapped and non-trapped trajectories on the x-z plane of derivatives. The scattered points are plotted on the x-z plane out of equally spaced time stepping.
\label{fig:trap2}}
\end{figure}

The scatter plot and 3D plot in Figure~\ref{fig:trap1} and \ref{fig:trap2} provide a comparison between trapped and non-trapped trajectories on the x-z plane. In Figure~\ref{fig:trap1}, the trajectory to the left is plotted out of initial $x=5.0, y=-5.0, z=12.0$, which suffers both high systematic loss and high informational loss. And the trajectory to the right is plotted out of initial $x=5.0, y=-12.0, z=12.0$, which suffers high systematic loss but little informational loss.

As a result, we observe drastic differences in learnability from these two initial settings, which is captured by comparing the l2 coefficient loss learned from same SINDy model after acquiring the same amount of training data generated with different initials.

\subsection{Stability of the information metric}

In order to further concretize the relationship between our metric and the information loss, we experiment 100 different initial conditions on the Lorenz system and provide scatter plots for their information metric vs L1 loss in the estimated coefficient, shown in Figure~\ref{fig:loss_dist}. The trajectory from each initial setting is fed into the SINDy model with two sets of training window, shown in the blue color (larger training data) and red color (less training data).

\begin{figure}[t]
\centering
\includegraphics[scale=0.5]{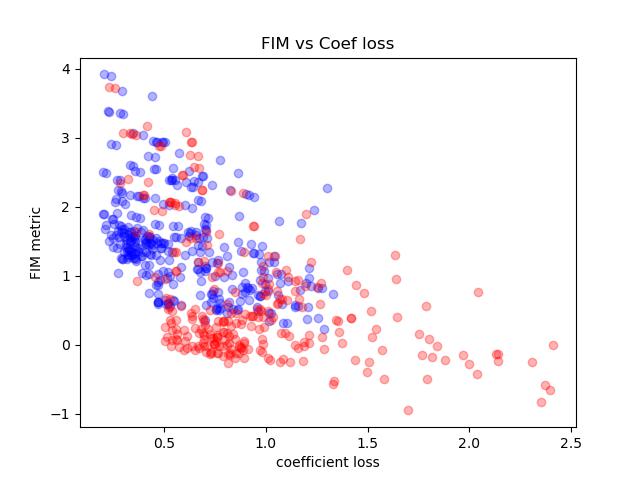}
\includegraphics[scale=0.5]{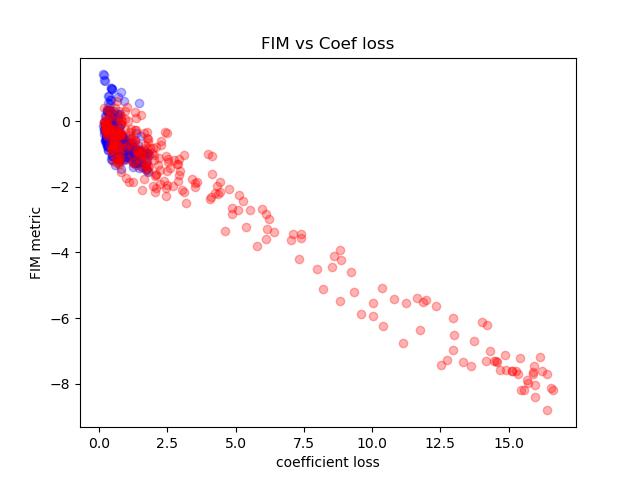}

\caption{Distribution of coefficient loss and information metrics for two combinations of training lengths. In each panel, the red group corresponds to the blue group at the specified fraction of the full training length: the left panel shows $\frac{1}{2}\times$, and the right panel shows $\frac{1}{8}\times$.
\label{fig:loss_dist}}
\end{figure}

With a prolonged training window, we observe a cluster of trajectories with similar FIM metric value and similarly small loss in coefficient. In other words, we observe little difference between initial settings after gathering a large enough set of training data. However, if we shorten the length of training window, we observe a roughly linear relationship between the metric value and the coefficient loss, where differences in learnability between initial trajectories start to be visible. Such pattern is more significant when we further shorten the training window, indicating better stability in small data scenarios.

\section{Concrete usages and performances under three scenarios}

In this section, we demonstrate concrete usages of the Fisher information and entropy metrics on promoting
data efficiency for an unknown dynamical system in three different cases - when only one trajectory is
available, when a control parameter is available for tuning, and when multiple trajectories are available
with freely chosen initials.

\subsection{On a single trajectory}

On the identification of nonlinear dynamical systems, while we could often generate inputs from repeated sampling process, in other practical settings, measurements could be expensive and trajectory might be unique. Without the option of ensembling, we could suffer from the locally flat parameter space of the FIM on a less informatory trajectory. Motivated by such scenarios, it is critical to for the identification framework to have as much exploration as possible on more informative data. Therefore, we propose techniques driven by information theory that both provide insight for the unknown system itself, and enable faster convergence in the case of "bad" initials.

\subsubsection{Interpretation from selective sampling}

Similar to the notion of self-supervised learning in robotics, some have taken an active learning approach on more complicated systems which require much larger feature spaces. This includes using learned model ensembles, whose usefulness are measured by the performances in the exploration stage \cite{shyam2019model,pathak2019self}, and using trajectory planning and tracking techniques in data collection for optimal design \cite{mania2020active}.

In data driven models, information analysis on the input data itself could also provide critical insights. In the scenario of selective sampling, we have access to all input data and try to select as few data as possible to achieve good performance. Through sampling on different trajectories, we could acquire further interpretation on the behavior of potentially unknown dynamical systems by learning the most informative portion of data from a series of measurements.

Taking the diffusion process as an example, as shown in Figure~\ref{fig:noise_masked}, for the reaction-diffusion system below with spiral waves on a periodic domain:

\begin{align}
&u_t = 0.1\nabla^2 u + (1-A^2)u +\beta A^2v\nonumber \\
&v_t = 0.1\nabla^2 v - \beta A^2 u + (1-A^2)v \\
&A^2 = u^2 + v^2. \nonumber
\end{align}

When performing optimization on the same initial condition as in the PySINDy example~\cite{kaptanoglu2021pysindy}, but using $50\%$ of the available training measurements (50 out of 100 timesteps on a $64\times 64$ grid), the resulting reconstruction is largely inadequate:

\begin{equation}
\begin{cases}
\displaystyle
\partial_t u = 0.784\, u + 0.220\, v - 0.781\, u^3 + 0.756\, v^3 - 0.784\, u v^2 + 0.754\, u^2 v + 0.082\, \nabla^2 u,\\
\displaystyle
\partial_t v = -0.258\, u + 0.783\, v - 0.715\, u^3 - 0.781\, v^3 - 0.713\, u v^2 - 0.780\, u^2 v + 0.080\, \nabla^2 v.
\end{cases}
\end{equation}

In contrast, when stochastic perturbations are introduced and masked in the initial condition, the modified initialization exhibits a substantial increase in information-theoretic metrics, including the largest eigenvalue and spectral skewness. Under these conditions, the model demonstrates improved performance even when trained on only $30\%$ of the measurements:

\begin{equation}
\begin{cases}
\displaystyle
\partial_t u = 1.009\, u - 1.004\, u^3 + 1.001\, v^3 - 1.015\, u v^2 + 0.995\, u^2 v + 0.097\, \nabla^2 u,\\[2mm]
\displaystyle
\partial_t v = 0.996\, v - 1.001\, u^3 - 0.996\, v^3 - 1.004\, u v^2 - 0.994\, u^2 v + 0.102\, \nabla^2 v.
\end{cases}
\end{equation}

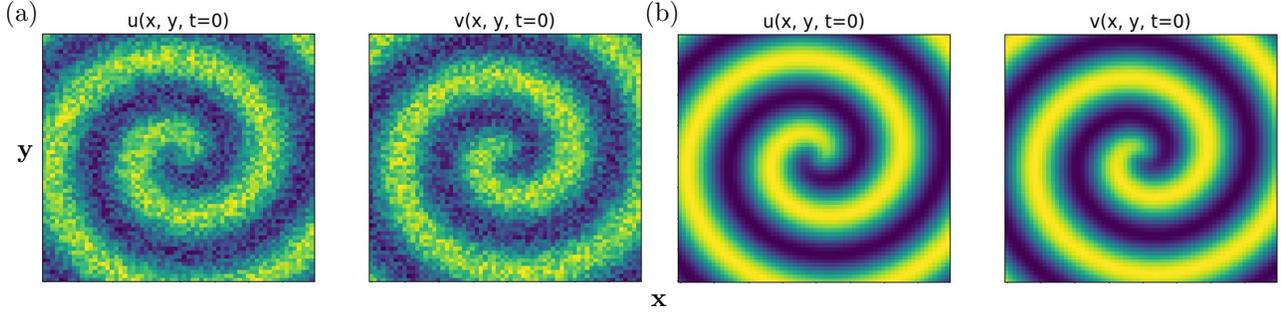
\begin{figure}[t]
\centering
\begin{overpic}[scale=0.42]{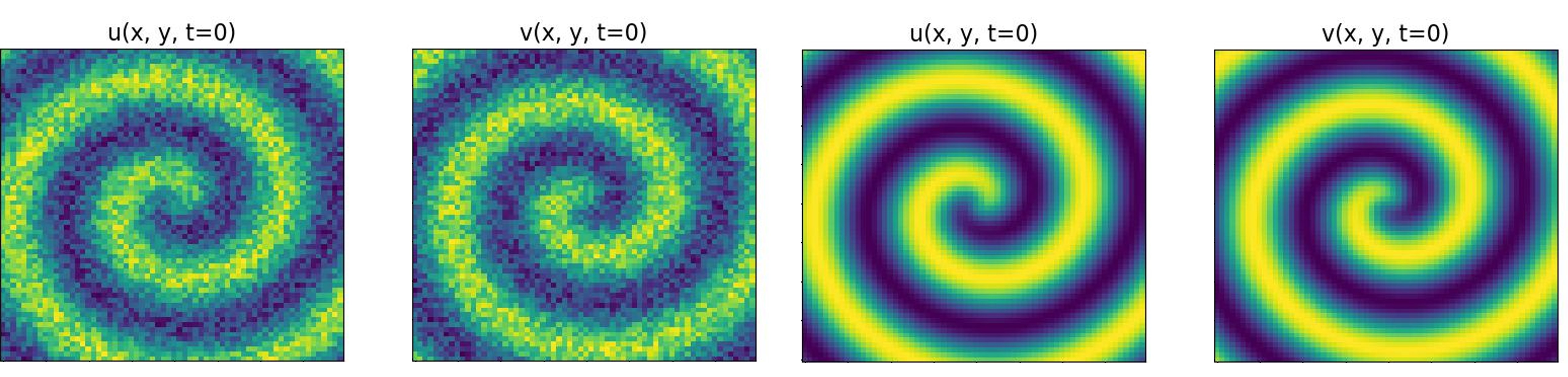}
    \put(-2,11){$\bf y$}
    \put(49,-1){$\bf x$}
    \put(-3,22){(a)}
    \put(48.5,22){(b)}
    \end{overpic}
\caption{The figure illustrates two initial conditions of the same reaction-diffusion system: (a) with masked noise in the initial condition, and (b) without noise in the initial condition, with (a) providing significantly more informative data than (b).
\label{fig:noise_masked}}
\end{figure}

This behavior is also observed in other diffusion-dominated dynamics: an initially smooth or well-structured condition frequently yields poorer early-time performance compared to a randomized or partially noisy initialization. Information-theoretic analysis suggests that the system is most effectively learned from initial conditions that are informative but not overly smoothed, thereby maximizing the entropy content of the training data. In other words, the salient features of a diffusive process are captured primarily through the system’s evolution as it organizes noise into coherent patterns, rather than from static pre-formed patterns. For example, in a reaction–diffusion system, initializing with a pre-existing spiral pattern impedes early learning relative to allowing a spiral structure to emerge dynamically from noisy data. Consequently, the evolution of information-rich features provides deeper insight into the diffusion process - and such features would generally differ for other dynamical systems, such as dispersive systems like the Schrödinger equation.

\subsubsection{Faster convergence rate from predictive sampling}

In data driven models, as is discussed in previous chapters, when a initial condition fails to provide an informative starting trajectory for certain complicated systems, we could no longer exploit the aggregated information gained from starting data for efficient parameter-finding. The information curve could be relatively flat on the input data, which prohibits us from directly observing the data blocks with the most information. For such scenario, we could either optimize the sampling approach to skip less informatory regions, or have perturbations on control parameters to gain more distortions in the parameter space.

The ultimate goal of predictive sampling is to reduce the number of measurements needed towards learning a dynamical system. In this scenario, we're trying to observe as few data as possible to provide candidate for a "good" initial condition with more information gained from the start of the trajectory. With a control parameter, we could actively learn and optimize the information gain for faster convergence, which will be further discussed in section 5.3. Without control parameters, we would be constrained onto a single trajectory and therefore it is critical to make measurements on the most informative parts of data. Thus, the predictive sampling method could arise as following, where $\beta > 1$ is a hyperparameter for sampling looseness, and $\alpha_1 > 0$ and $\alpha_2 < 0$ are thresholds for information burst or implosion:

\begin{algorithm}[t]
\caption{Adaptive Predictive Sampling for Single Trajectory Model Discovery}
\label{alg:predictive_sampling}
\textbf{Input:} Initial condition $\mathbf{x}_0$; unknown dynamics $f$; base sampling interval $\delta t$; looseness factor $\beta > 1$; information metric $\mathcal{I}(\cdot)$; burst threshold $\gamma_{\mathrm{up}} > 1$; collapse threshold $\gamma_{\mathrm{down}} \in (0,1)$; minimum data size $N_{\min}$; maximum iterations $N_{\max}$

\textbf{Output:} Training dataset $\mathcal{D}$; learned model $\hat{f}$

\begin{algorithmic}[1]
\State Initialize $\mathcal{D} \gets \{(t_0, \mathbf{x}_0)\}$, $t \gets 0$, $n \gets 1$
\State $\mathcal{I}_{\mathrm{prev}} \gets 0$ \Comment{Previous information metric value}
\State $\texttt{mode} \gets \texttt{COARSE}$ \Comment{Start with coarse sampling}

\While{$n < N_{\max}$}
    \If{$\texttt{mode} = \texttt{COARSE}$}
        \State $t \gets t + \beta \cdot \delta t$ \Comment{Loose sampling interval}
    \Else
        \State $t \gets t + \delta t$ \Comment{Tight sampling interval}
    \EndIf
    
    \State Observe $\mathbf{x}_n \gets \mathbf{x}(t)$ from dynamics $f$
    \State $\mathcal{D} \gets \mathcal{D} \cup \{(t, \mathbf{x}_n)\}$
    \State Fit model: $\hat{f} \gets \texttt{SINDy}(\mathcal{D})$
    \State Compute $\mathcal{I}_n \gets \mathcal{I}(\mathcal{D}, \hat{f})$ \Comment{e.g., $\lambda_1(\mathbf{I})$ or combined metric}
    
    \If{$\texttt{mode} = \texttt{COARSE}$ \textbf{and} $\mathcal{I}_n > \gamma_{\mathrm{up}} \cdot \mathcal{I}_{\mathrm{prev}}$}
        \State $\texttt{mode} \gets \texttt{FINE}$ \Comment{Information burst detected}
    \ElsIf{$\texttt{mode} = \texttt{FINE}$ \textbf{and} $\mathcal{I}_n < \gamma_{\mathrm{down}} \cdot \mathcal{I}_{\mathrm{prev}}$}
        \State $\texttt{mode} \gets \texttt{COARSE}$ \Comment{Information collapse detected}
    \EndIf
    
    \State $\mathcal{I}_{\mathrm{prev}} \gets \mathcal{I}_n$
    \State $n \gets n + 1$
    
    \If{$|\mathcal{D}| \geq N_{\min}$ \textbf{and} model quality criterion satisfied}
        \State \textbf{break}
    \EndIf
\EndWhile

\State \Return $\mathcal{D}$, $\hat{f}$
\end{algorithmic}
\end{algorithm}

Intrinsically we only start tight sampling until seeing a burst in the information metric, and stop on the observation of a collapse. In practice, we may also adjust the stopping criteria to accomodate for continuous increase or descent.

\subsubsection{On an active initial trajectory}

From the results in the previous sections, it is straightforward to observe that under certain initial conditions, the rise in information starts from the very beginning of the trajectories. Such phenomenon becomes more significant for initial settings that have longer transient stage towards oscillations at the start of the trajectory. In such cases, the algorithm would acquire training data in a tightly manner from the initial trajectory and give satisfiable results, as shown for the example of the Lorenz equation in Figure~\ref{fig:lor_tractor}.

\paragraph{Example 1: Lorenz system \\ \\}

\begin{figure}[t]
\centering
\begin{overpic}[scale=0.48]{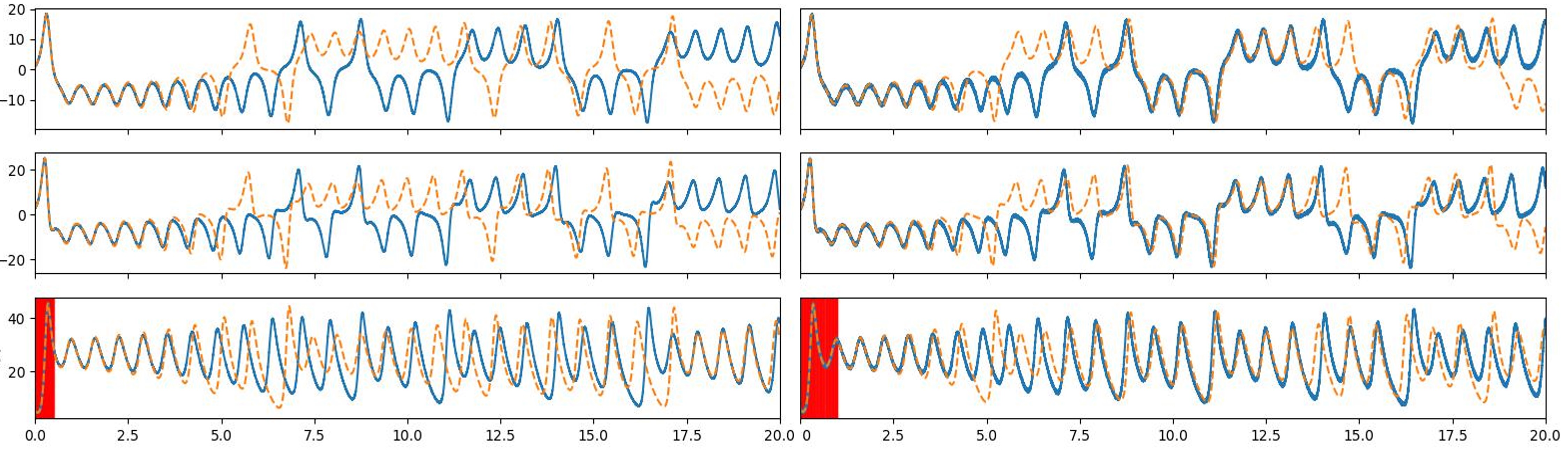}
    \put(-2.5,24.5){\footnotesize$\bf x(t)$}
    \put(-2.5,15){\footnotesize$\bf y(t)$}
    \put(-2.5,6){\footnotesize$\bf z(t)$}
    \put(48,-1.5){Time}
     \put(99,24.5){\footnotesize$\bf x(t)$}
    \put(99,15){\footnotesize$\bf y(t)$}
    \put(99,6){\footnotesize$\bf z(t)$}
    \put(-2.5,29){(a)}
    \put(49,29){(b)}
    \end{overpic}
\caption{Two models are trained on the Lorenz system using the same initial condition located far from the attractor, but with varying lengths of training datasets. Panel (a) shows the case where only data leading up to (but not including) the first oscillation are used for training. Panel (b) shows the case where the training data include both the transient leading to and the first oscillation itself.
\label{fig:lor_tractor}}
\end{figure}
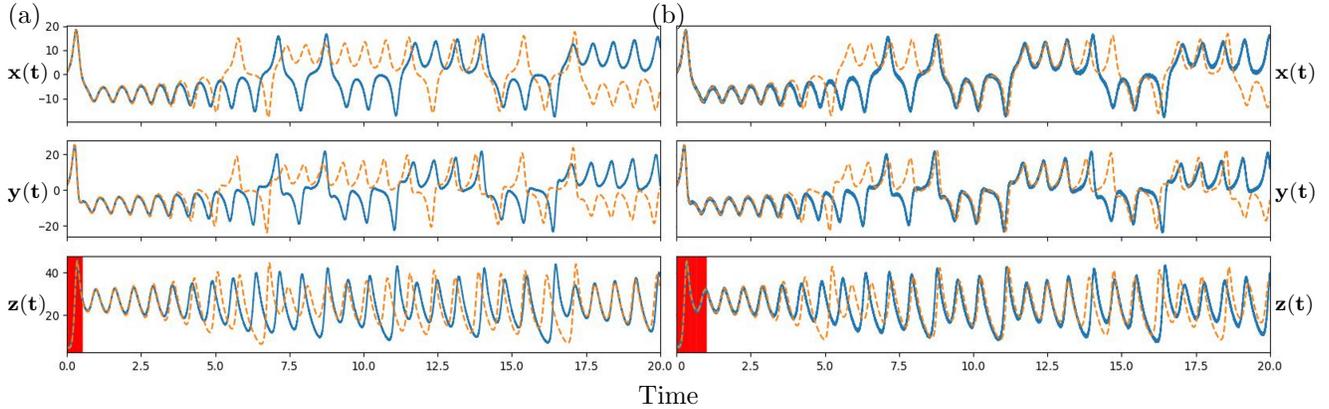

To reveal how discriminative sampling could affect model performance, Figure~\ref{fig:noise_resist} compares the performances of SINDy with different sampling strategies of data for two sets of initial conditions of the Lorenz system. As revealed from the graphs, the data leading to the first cycle could provide information that yields better results than a random selection of the same amount of data from later oscillations. And at different noise levels, the SINDy framework on starting data gives better performance than a random selection of more data. When only data leading to and including the first oscillation are used for training, the method could stably withstand a noise level higher than the case when only data leading to but excluding the first oscillation are seen. 

However, another important metric to consider that is not shown from the graph would be the number of outliers. Compared with an average of 30 percent outliers when seeing only the starting 10 percent of data, the amount of outliers drastically reduces to well below 10 percent when the training data includes the first oscillation.

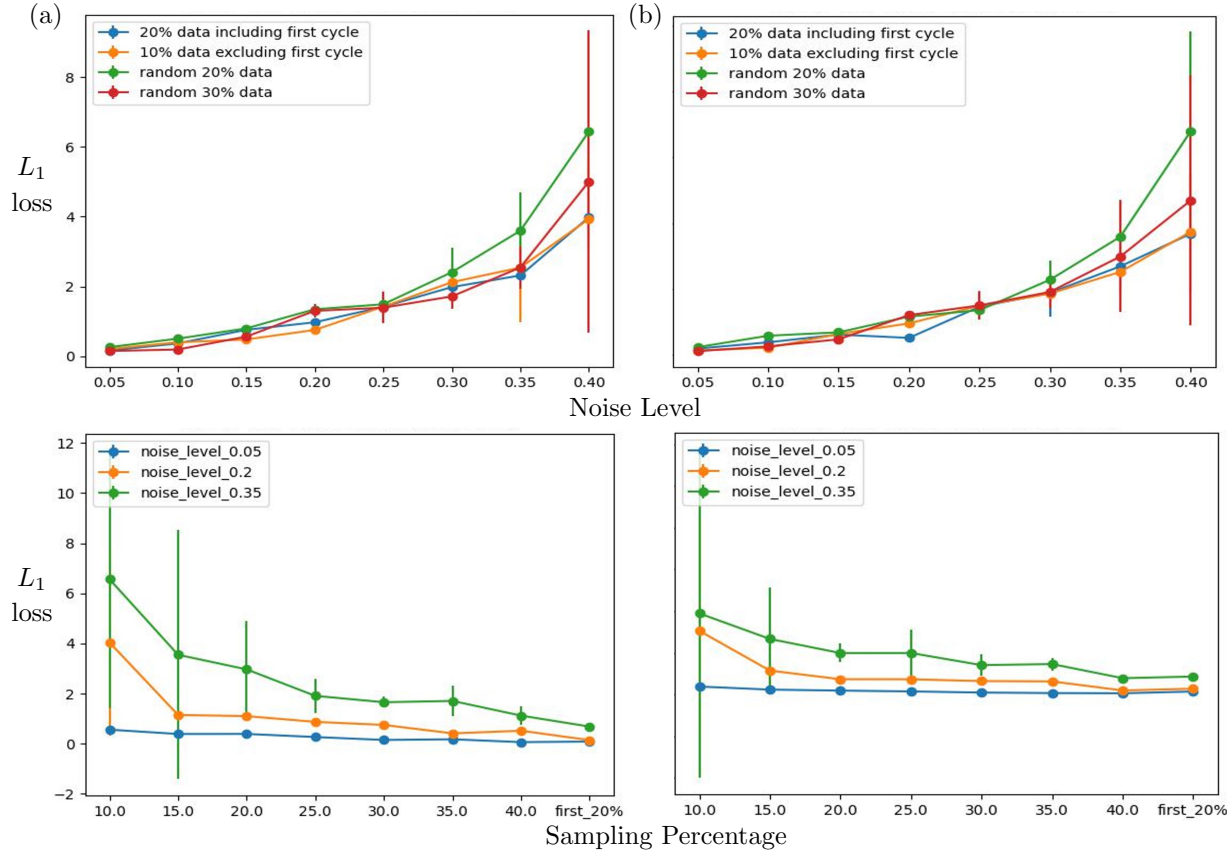
\begin{figure}[t]
\centering
\begin{overpic}[scale=0.58]{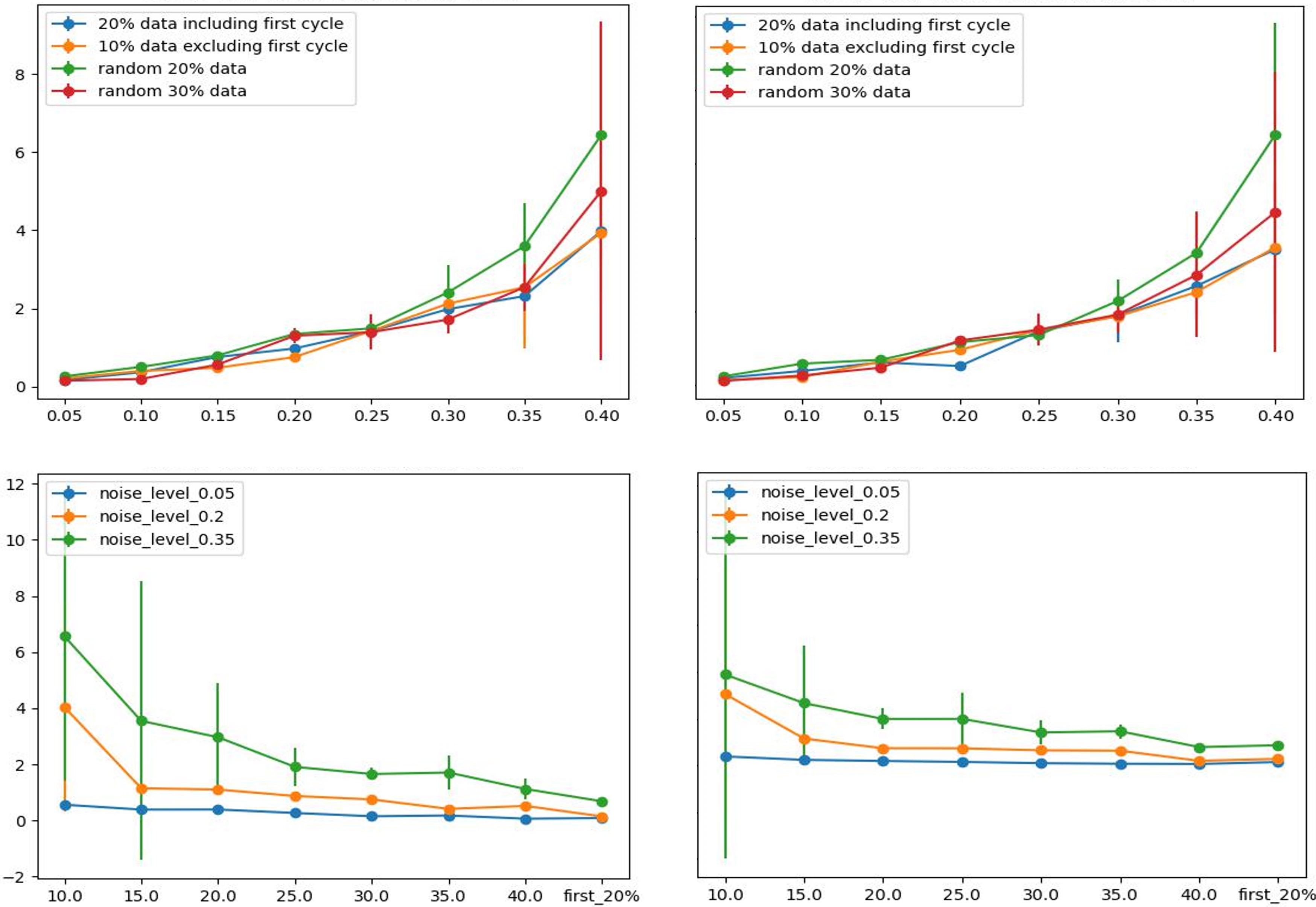}
    \put(44,34.5){Noise Level}
    \put(42,-2){Sampling Percentage}
    \put(-2,68){(a)}
    \put(49,68){(b)}
    \put(-3,20){$L_1$}
    \put(-3.5,17){loss}
    \put(-3,55){$L_1$}
    \put(-3.5,52){loss}
    \end{overpic}
    \vspace{1em}
\caption{Top panels: mean and variance of $L_1$ coefficient loss versus noise levels for SINDy trained on selected data samplings of the Lorenz system. Bottom panels: mean and variance of $L_1$ coefficient loss versus sampling percentages at selected noise levels. Left panels (a) correspond to initial condition $[1,3,5]$; right panels (b) correspond to initial condition $[3,3,5]$. A few outliers were automatically excluded.
\label{fig:noise_resist}}
\end{figure}

\paragraph{Example 2: Rössler attractor \\ \\} 

Applying the same algorithm on the Rössler attractor with coefficients a=0.2, b=0.2 and c=5.7 (where it has chaotic behavior) yields similar results, as shown in Figure~\ref{fig:pred_samp}(a). At various noise levels, predictive sampling method would yield data from the start of the trajectory, which outperforms a random selection of input data.

\begin{figure}[t]
\centering
\begin{minipage}{0.72\linewidth}
  \begin{subfigure}{\linewidth}
    \begin{overpic}[width=\linewidth]{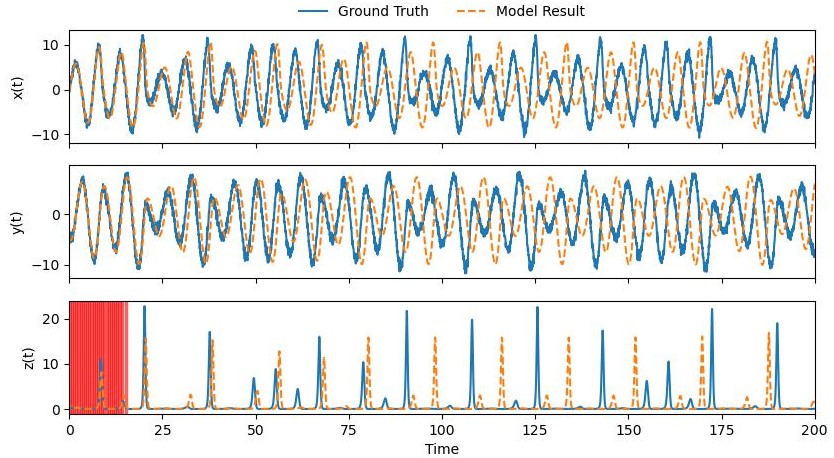}
    \put(2,55){\textbf{(a)}}
    \end{overpic}
    \caption*{}\label{subfig:key-a}
  \end{subfigure}
\end{minipage}
\hfill
\begin{minipage}{0.25\linewidth}
  \begin{subfigure}{\linewidth}
  \begin{overpic}[scale=0.4]{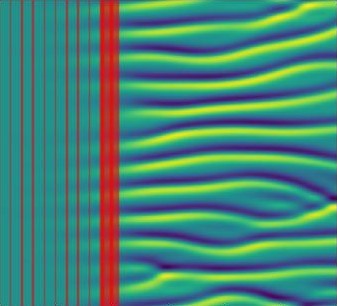}
  \put(34,95){$\bf v(x,t)$}
  \put(45,-8){$\bf t$}
    \put(-10,40){$\bf x$}
    \put(-15,98){\textbf{(b)}}
    \end{overpic}
    \caption*{}\label{subfig:key-b}
  \end{subfigure}

  \medskip

  \begin{subfigure}{\linewidth}
    \begin{overpic}[scale=0.4]{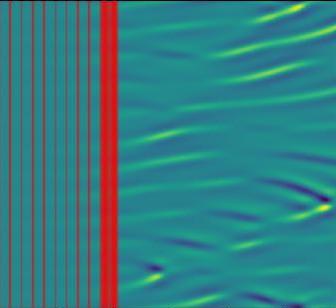}
    \put(34,95){$\bf \dot{v}(x,t)$}
    \put(45,-8){$\bf t$}
    \put(-10,40){$\bf x$}
    \end{overpic}
    \caption*{}\label{subfig:key-c}
  \end{subfigure}
\end{minipage}

\caption{The result of predictive sampling method applied on Rossler attractor (a) and  Kuramoto–Sivashinsky (KS) equation (b). The method would quickly skip the less informative region by collecting data only from the highlighted regions.
\label{fig:pred_samp}}
\end{figure}

\subsubsection{On a silent initial trajectory}

A good example where the dynamics remain silent at the starting trajectory could be the famous Kuramoto–Sivashinsky (KS) equation\cite{ashinsky1988nonlinear}, one of the simplest canonical partial differential equation models capable of generating rich spatio-temporal chaos, despite its deceptively compact form. Originally derived in the study of diffusive instabilities in hydrodynamic flame fronts, the KS equation features linear growth, nonlinear advection and fourth-order dissipation that lead to chaotic dynamics even in one spatial dimension.~\cite{hyman1986kuramoto}

\begin{align}
v_t + v_{xx} + v_{xxxx} + vv_x = 0
\end{align}

We use it as a testbed for trajectories with silent start from a random initialization, as shown in Figure~\ref{fig:pred_samp}(b) below. The predictive sampling method would quickly skip the less informative data by collecting data only from the highlighted regions, before observing a spike of increase in the information metric.

Through this approach, we encouraged the model to focus on information-optimized segments of the trajectory, enabling effective learning with substantially fewer measurements. In experiments on the Kuramoto–Sivashinsky and reaction–diffusion systems, the number of required measurements was reduced to less than one quarter of the original amount while still achieving successful identification of the underlying dynamics.

A similar strategy is expected to be applicable to the FitzHugh–Nagumo model~\cite{rocsoreanu2012fitzhugh}, where informative transient dynamics can likewise be exploited for efficient system identification.

\subsection{SINDy with control}

In the field of control theory, Eurika Kaiser et al \cite{kaiser2018sparse} provides good indication of how perturbations could be applied for quick learning in multiple dynamical systems. This could be extremely useful on scenarios that requires fast recovery from unknown dynamics, such as self-protection systems of motors to avoid crashes. When the Koopman operator representation of a quadcopter model is actively learned, the recovery time after one of the rotors fails could be drastically reduced.\cite{taylor2021active}

In terms of information theory under the SINDy framework, we could also actively learn and optimize the Fisher information at each step with respect to an unknown model. In an experiment on the Lorenz system, a steady perturbation parameter is kept until the model gets good prediction horizon feedback. Then, until stopping criteria on prediction quality is met, the perturbation parameter is iteratively updated based from info-optimization on a combination of new data and simulations of learned models in the past. Such design clears the need for parallel measurements, i.e. multiple perturbations exerted on the system at the same moment. A 2-percent noise is included, and Figure~\ref{fig:control} presents the model results compared with baselines from Eurika's paper \cite{kaiser2018sparse} in terms of coefficient error and prediction error.

\begin{figure}[t]
\centering
\begin{overpic}[scale=0.58]{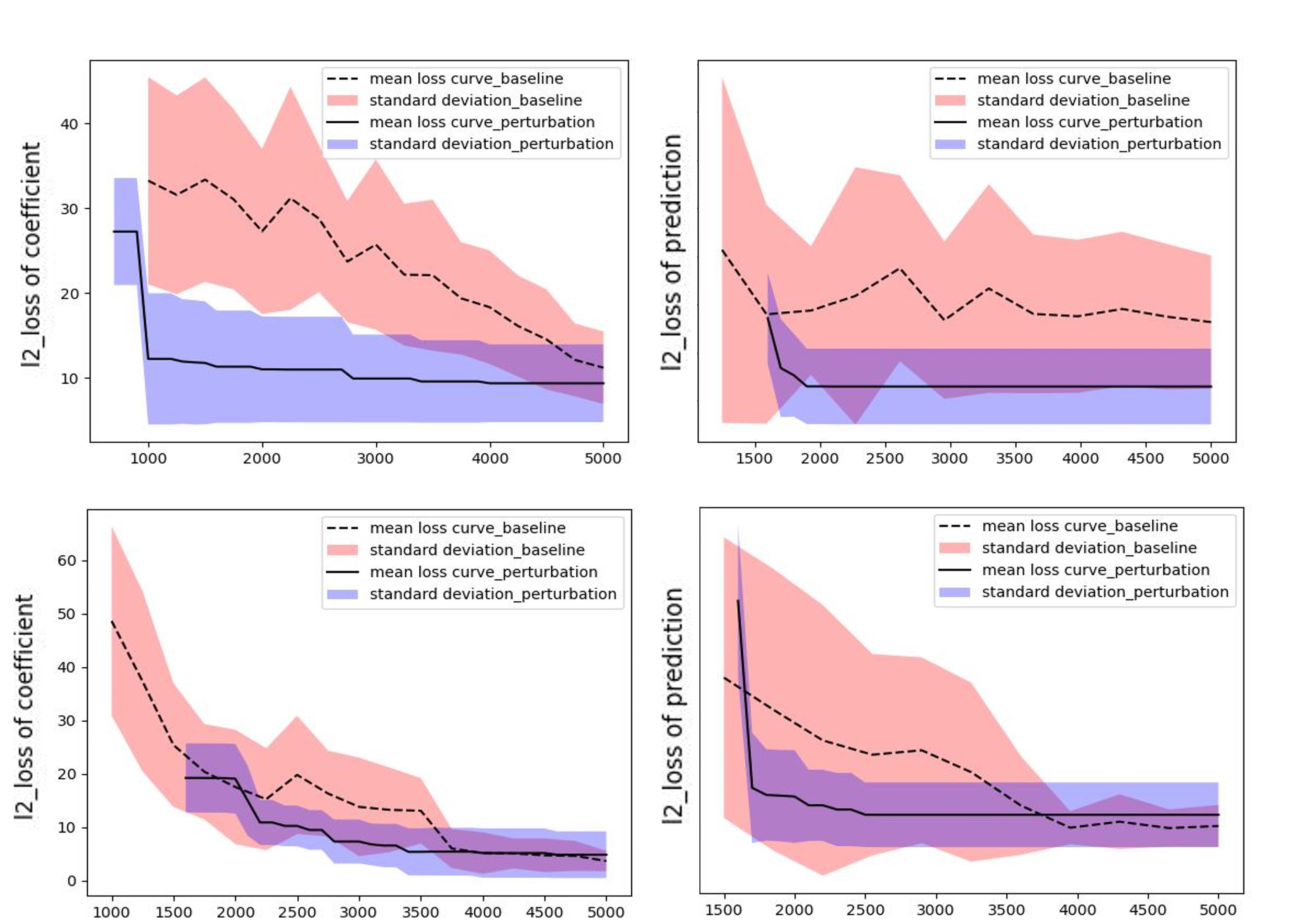}

    \put(42,-2){Training data length}
    \put(2,65){(a)}
    \put(2,33){(b)}
    \end{overpic}
        \vspace{1em}
\caption{Model results for active information-driven SINDy with control on the Lorenz system with a 2\% perturbation noise. Panel (a) corresponds to an initial condition close to the orbit, and Panel (b) corresponds to an initial condition farther from the orbit. The perturbation parameter is iteratively optimized based on Fisher information from new data and previously learned simulations until the prediction quality stopping criteria is met. Performance is compared with baselines from Eurika et al.~\cite{kaiser2018sparse}. The control parameter converges quickly, significantly accelerating model discovery for the less informative initial condition (a), while improvements are more modest for the richer initial condition (b).
\label{fig:control}}
\end{figure}

It could be observed that the control parameter converges quickly within a few iterations to meet the stopping criteria. It significantly accelerates the model discovery process when the initial condition is close to the orbit - a less informative initial condition discussed in previous parts of this paper. But when the initial condition is further off the orbit, model performance is improved by a relatively smaller amount in both metrics, due to the rich of information from the initial trajectory itself.

\subsection{Entropy search SINDy}

\begin{figure}[t]
\centering
\includegraphics[scale=0.5]{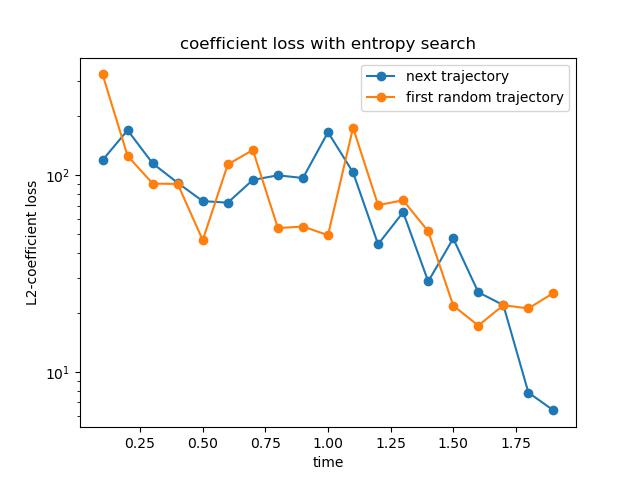}
\includegraphics[scale=0.5]{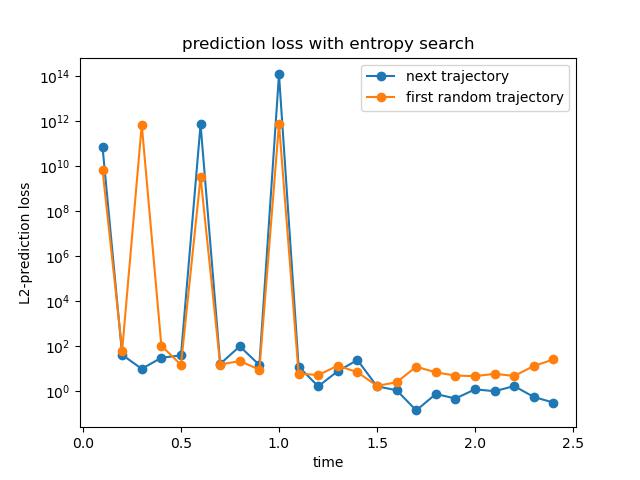}

\caption{The reduction in L2-coefficient loss and prediction loss when going from a randomly initialized trajectory to a more informative trajectory
\label{fig:search}}
\end{figure}

When initial trajectories can be freely placed on the domain for data collection, it would be a lot more efficient to search for a more informative initial trajectories, in contrast to randomly choosing the next ones.

As introduced in section 2.3.5, the family of entropy-based acquisition functions have been used in Bayesian statistical methods for more efficient and robust choices of evaluation points. In one of the recent works, Fröhlich et al\cite{frohlich2020noisy} proposed the noisy-input version of the entropy search function that account for not only measurement noises, but input noises as well.

Characterizing the perturbations on the input parameter with $p(\xi)\sim N(0, \Omega_x)$, the robust objective then becomes

\begin{equation}
    g(x) = E_{\xi\sim p(\xi)} [f(x + \xi)]
\end{equation}

instead of the original unknown objective f. And the acquisition function then is aimed at maximizing the mutual information between x and the robust objective's maximum value $g^* = max_{x\in \mathcal{X}} g(x)$, and therefore becomes

\begin{equation}
    \alpha_{NES}(x) = H[p(y|x,Data)]- E_{p(g^*|Data)}[H[p(y|x,g^*,Data)]]
\end{equation}

where the first term is the same as in the predictive entropy search (PES) method discussed in section 2.3.5, and can therefore be computed as the entropy of the posterior from a Gaussian process prior. And in the second term, the expectation over robust maximum can be approximated via any Monte Carlo sampling technique, and the posterior conditioned on the optimum can still be efficiently approximated by the expectation propagation method.

\begin{algorithm}[t]
\caption{General Bayesian Optimization}
\begin{algorithmic}
\State \textbf{Input:} a black box with an unknown objective function $f$

\State \textbf{For} $n = 1,2,\ldots,N$:
\begin{list}{}{\leftmargin=2em}  
    \item Select $x_n = \arg\max_{x \in \mathcal{X}} \alpha_{n-1}(x)$, where $\alpha_{n-1}(x)$ is the acquisition function for the next evaluation
    \item Use the Bayesian surrogate model at $x_n$ to approximate $y_n$
    \item Augment data: $U_n = U_{n-1} \cup \{(x_n, y_n)\}$
\end{list}

\end{algorithmic}
\end{algorithm}

\begin{algorithm}[t]
\caption{Entropy Search SINDy for Multi Trajectory Model Discovery}
\label{alg:entropy_search_sindy}
\textbf{Input:} Domain $\mathcal{X}$ for initial conditions; unknown dynamics $f$; number of trajectories $M$; samples per trajectory $N$; FIM-based metric $\mathcal{M}(\cdot)$; Gaussian process prior $\mathcal{GP}$

\textbf{Output:} Aggregate dataset $\mathcal{D}$; learned model $\hat{f}$
\begin{algorithmic}[1]
\State Sample initial condition $\mathbf{x}_1^{(1)} \sim \mathrm{Uniform}(\mathcal{X})$
\State Initialize trajectory $\mathcal{D}_1 \gets \{(t_0, \mathbf{x}_1^{(1)})\}$
\State Initialize aggregate dataset $\mathcal{D} \gets \mathcal{D}_1$
\For{$m = 1, 2, \ldots, M$}
    \State \textbf{// Phase 1: Temporal sampling within trajectory $m$}
    \For{$n = 1, 2, \ldots, N$}
        \State Update surrogate: $\mathcal{GP} \gets \texttt{FitGP}(\mathcal{D})$
        \State Compute acquisition: $\alpha_{n-1}(t) \gets H[p(\mathbf{y}|t, \mathcal{D})] - \mathbb{E}_{p(\boldsymbol{\xi}|\mathcal{D})}[H[p(\mathbf{y}|t, \boldsymbol{\xi}, \mathcal{D})]]$
        \State Select next time point: $t_n \gets \arg\max_{t \in \mathcal{T}} \alpha_{n-1}(t)$
        \State Observe $\mathbf{x}_n^{(m)} \gets \mathbf{x}^{(m)}(t_n)$ from dynamics $f$
        \State $\mathcal{D}_m \gets \mathcal{D}_m \cup \{(t_n, \mathbf{x}_n^{(m)})\}$
        \State $\mathcal{D} \gets \mathcal{D} \cup \{(t_n, \mathbf{x}_n^{(m)})\}$
    \EndFor
    
    \State \textbf{// Phase 2: Select next initial condition}
    \If{$m < M$}
        \State Train ensemble model: $\hat{f} \gets \texttt{Ensemble-SINDy}(\mathcal{D})$
        \State Compute aggregate FIM: $\mathbf{I}_{\Sigma} \gets \sum_{k=1}^{m} \mathbf{I}_k$
        \State Define acquisition for initial conditions:
        \Statex \quad \quad \quad \quad $\beta_{m}(\mathbf{x}_0) \gets \mathbb{E}_{\hat{f}}\left[\mathcal{M}(\mathbf{I}_{\Sigma} + \mathbf{I}_{\mathbf{x}_0}) - \mathcal{M}(\mathbf{I}_{\Sigma})\right]$
        \Statex \quad \quad \quad \quad where $\mathbf{I}_{\mathbf{x}_0}$ is the predicted FIM contribution from trajectory starting at $\mathbf{x}_0$
        \State Select next initial: $\mathbf{x}_1^{(m+1)} \gets \arg\max_{\mathbf{x}_0 \in \mathcal{X}} \beta_m(\mathbf{x}_0)$
        \State Initialize: $\mathcal{D}_{m+1} \gets \{(t_0, \mathbf{x}_1^{(m+1)})\}$
    \EndIf
\EndFor
\State \Return $\mathcal{D}$, $\hat{f} \gets \texttt{Ensemble-SINDy}(\mathcal{D})$
\end{algorithmic}
\end{algorithm}

The elegance of the entropy-based Bayesian model lies in the closed-format distribution functions for both priors and posteriors of the Gaussian process, especially when Gaussian noises are considered. And the limitations are also noticeable - Firstly, the input noise is assumed to be normally or uniformly distributed in order to have closed form kernel functions, so that more extreme cases (where some inputs are significantly more corrupt than others, i.e. their values are much more perturbed) are not considered. Secondly, the computational cost greatly increases for higher dimensional model, i.e. in the Hartmann 6-Dimensional function.

Sparsity-promoting algorithms like SINDy could account for the drawbacks above with its noise-robustness. And since entropy based acquisition functions can be used on the next evaluation point to maximize mutual information, it can also be used on searching for a new trajectory to maximize structural similarities. Fig~\ref{fig:search} shows the reduction in loss when a single step is performed (i.e. going from a randomly initialized trajectory to a slightly more informative trajectory)

\section{Conclusion and Discussions}

In this paper, we leverage the Fisher Information Matrix (FIM) within the data-driven framework of sparse identification of nonlinear dynamics (SINDy) to systematically quantify the informativeness of measurement data. By analyzing the FIM associated with both single trajectories and ensembles of trajectories, we demonstrate that information-theoretic analysis can significantly improve sampling efficiency and model performance by prioritizing data drawn from the most informative regions of state space and time. Furthermore, we show that Fisher information and entropy-based metrics can promote data efficiency for unknown dynamical systems across a range of practical scenarios, including cases where only a single trajectory is available, where a control parameter can be tuned, and where multiple trajectories with freely chosen initial conditions can be collected. These quantifiable information metrics enable more efficient sampling strategies that accelerate convergence while maintaining robustness, illustrating how information theory can be directly and effectively integrated into data-driven model discovery.

Several important directions remain for future investigation. Firstly, the determinant of the Fisher Information Matrix can be interpreted as the inverse determinant of the parameter covariance matrix, commonly referred to as the generalized variance \cite{wilks1932certain}. Within the framework of D-optimality, this quantity admits a clear geometric interpretation in terms of the volume of the parameter confidence ellipsoid and has been linked to model complexity and identifiability \cite{ly2017tutorial}. Incorporating an explicit penalization of model complexity into information-theoretic metrics therefore represents a promising direction, and systematic numerical experiments could be conducted to assess its effect on model selection, robustness, and overfitting behavior.

Secondly, the optimization of control parameters explored in Section 5.2 is driven purely by information-theoretic objectives. A complementary and widely used approach is model predictive control (MPC) \cite{lee2011model, mayne2014model}, which optimizes control inputs over a finite prediction horizon and has seen increasing adoption in nonlinear and learning-based settings \cite{bieker2020deep}. Integrating information-based criteria with MPC may yield further performance gains. For example, information metrics could be incorporated directly into the MPC cost function \cite{kaiser2018sparse}; however, enforcing convexity typically requires restricting attention to scalar summaries of the FIM, such as its largest eigenvalue. Leveraging the full spectral structure of the FIM would instead lead to nonconvex optimization and potential convergence to local minima. An alternative strategy is to evaluate multiple candidate forcing inputs at each timestep and aggregate them using weights informed by predictive performance, information gain, and model complexity.

Thirdly, in tasks involving neural networks - such as reconstruction problems under sparse sensing - recent studies suggest that random sensor placement may not significantly degrade reconstruction accuracy \cite{williams2024sensing}. This behavior may arise from redundancy in the latent space or from the availability of sufficiently rich time-delay embeddings despite a reduced number of sensors. Understanding how such redundancy interacts with information-theoretic sensing and sampling criteria remains an open question and may help clarify when principled sensor placement provides advantages over random strategies.

In the present study, Fisher information and entropy-based metrics can be used to retrospectively identify informative trajectories, temporal segments, and initial conditions, enabling substantial gains in data efficiency. In sensor networks, these ideas have been developed into sensor management frameworks that leverage entropy or mutual information for sensor selection and placement to maximize expected reduction in uncertainty and improve parameter inference~\cite{krause2008near,lee2021optimal}.

Therefore, an important direction for future work is the integration of information-guided sampling with adaptive, closed-loop sensing strategies on data assimilation models. Recent work on DA-SHRED~\cite{bao2025data} provides a foundation by combining data assimilation with learned latent representations for efficient reconstruction and model discovery in high-dimensional, spatio-temporal systems. While DA-SHRED currently operates under prescribed sensing layouts, its architecture naturally lends itself to the incorporation of information-theoretic criteria for adaptive sensor placement. Embedding Fisher information–based discriminative metrics into such data-assimilation frameworks represents a promising pathway toward fully autonomous, information-optimal model discovery under severe sensing constraints.

\section*{Acknowledgements}

We are especially indebted to Prof. Urban Fasel for invaluable conversations regarding the outlook and prospectus of this work.
This work was supported in part by the US National Science Foundation (NSF) AI Institute for Dynamical Systems (dynamicsai.org), grant 2112085. JNK further acknowledges support from the Air Force Office of Scientific Research  (FA9550-24-1-0141).

\printbibliography{}

\end{document}